\documentclass[a4paper,11pt]{article}

\usepackage{fullpage}
\usepackage{times}
\usepackage{soul}
\usepackage{url}
\usepackage[utf8]{inputenc}
\usepackage[small]{caption}
\usepackage{graphicx}
\usepackage{xcolor}
\usepackage{amsmath}
\usepackage{booktabs}
\usepackage{enumerate}
\usepackage{dsfont}

\usepackage[ruled,linesnumbered,vlined]{algorithm2e}
\SetKwRepeat{Do}{do}{while} 

\SetCommentSty{mycommentfont}

\usepackage{amsthm}
\usepackage{amssymb}
\usepackage{amsfonts}
\usepackage{algpseudocode}

\usepackage{bbm}
\usepackage{mathtools}
\usepackage{multirow}
\usepackage{multicol}
\usepackage{subcaption}
\usepackage{cleveref}

\usepackage{tikz}
\usetikzlibrary{positioning}
\tikzset{main node/.style={circle,fill=white,draw,minimum size=1.2cm,inner sep=0pt},}
\tikzset{text/.style={none,fill=white,draw,minimum size=1cm,inner sep=0pt},}

\newtheorem{theorem}{Theorem}[section]
\newtheorem{corollary}[theorem]{Corollary}
\newtheorem{proposition}[theorem]{Proposition}
\newtheorem{lemma}[theorem]{Lemma}
\newtheorem{observation}[theorem]{Observation}

\theoremstyle{definition}
\newtheorem{definition}[theorem]{Definition}

\DeclarePairedDelimiter\floor{\lfloor}{\rfloor} 
\DeclareMathOperator*{\argmax}{arg\,max}

\newcommand{\egalOPT}{\textup{OPT-egal}}
\newcommand{\egalSW}{\textup{SW-egal}}
\newcommand{\utilOPT}{\textup{OPT-util}}
\newcommand{\utilSW}{\textup{SW-util}}
\newcommand{\egalPrice}{\textup{Egal-PoC}}
\newcommand{\utilPrice}{\textup{Util-PoC}}

\newcommand{\alloc}{\mathcal{M}}
\newcommand{\connectedAllocsOf}[1]{C(#1)}
\newcommand{\utilityProfile}{\mathcal{U}}

\usepackage{natbib}
\usepackage{authblk}

\allowdisplaybreaks

\title{\bf Welfare Loss in Connected Resource Allocation}

\author[1]{Xiaohui Bei}
\author[2]{Alexander Lam}
\author[3]{Xinhang Lu}
\author[4]{Warut Suksompong}

\affil[1]{Nanyang Technological University, Singapore}
\affil[2]{Hong Kong Polytechnic University, Hong Kong}
\affil[3]{University of New South Wales, Australia}
\affil[4]{National University of Singapore, Singapore}

\date{\vspace{-10mm}}

\begin{document}

\maketitle

\begin{abstract}
We study the allocation of indivisible items that form an undirected graph and investigate the worst-case welfare loss when requiring that each agent must receive a connected subgraph.
Our focus is on both egalitarian and utilitarian welfare.
Specifically, we introduce the concept of \emph{egalitarian (resp., utilitarian) price of connectivity}, which captures the worst-case ratio between the optimal egalitarian (resp., utilitarian) welfare among all allocations and that among connected allocations.
We provide tight or asymptotically tight bounds on the price of connectivity for several large classes of graphs in the case of two agents---including graphs with vertex connectivity $1$ or $2$ and complete bipartite graphs---as well as for paths, stars, and cycles in the general case where the number of agents can be arbitrary.
\end{abstract}

\section{Introduction}

The division of resources among interested agents is a common problem in everyday life, and has been extensively studied in the fields of mathematics, economics, and computer science \citep{BramsTa96,RobertsonWe98,Moulin03}.
Several applications of this problem involve allocating \emph{indivisible items}, such as distributing artwork among museums or assigning volunteers to community service clubs.
Taking a \emph{utilitarian} standpoint, one may assign each item to an agent with the highest utility for it, so as to maximize the sum of the agents' utilities.
However, this can lead to inequitable allocations where certain agents are left with little or no utility.
For example, suppose that Alice has the same value for all items, while the other agent Bob has slightly higher value than Alice for every item except the last one, which he values much less than her.
Then, the utilitarian-optimal allocation assigns only the last item to Alice, which is arguably unfair to her.
On the other hand, a fairer allocation may be achieved by an \emph{egalitarian} approach, which maximizes the smallest utility across all agents.

In practice, there are often constraints on the feasible allocations.\footnote{We refer to the survey by \citet{Suksompong21} for an overview of constraints in resource allocation.}
One of the most important types of constraints is \emph{connectivity}, which was initially considered in the context of indivisible items by \citet{BouveretCeEl17}, and subsequently studied by a number of authors (see \Cref{sec:related-work} for further discussion of related work).
Connectivity constraints arise naturally, for example, when allocating offices in a university building to research groups---a connected set of offices facilitates communication among members of the same group---as well as when dividing retail units in a shopping mall among retailers.
Formally, the items correspond to the vertices of a connected (undirected) graph, and the bundle of items assigned to each agent must form a connected subgraph of this graph.
As one may expect, imposing connectivity requirements can lead to a loss in social welfare when compared to the optimal unconstrained welfare.
In this paper, we therefore address the following research question:
\begin{quote}
\itshape
When allocating indivisible items corresponding to vertices of a graph, what is the worst-case loss of social welfare resulting from connectivity constraints?
\end{quote}

\begin{figure}[t]
\centering
\begin{tikzpicture}
\node[main node] (1) at (-2,1) {};
\node[main node] (2) at (2,1)  {};
\node[main node] (3) at (-2,-1) {};
\node[main node] (4) at (2,-1) {};
\node[draw=none] (1a) at (-2,1.2) {$1$: $0.5$};
\node[draw=none] (1b) at (-2,0.8) {$2$: $0.2$};
\node[draw=none] (2a) at (2,1.2) {$1$: $0.0$};
\node[draw=none] (2b) at (2,0.8) {$2$: $0.4$};
\node[draw=none] (3a) at (-2,-0.8) {$1$: $0.4$};
\node[draw=none] (3b) at (-2,-1.2) {$2$: $0.0$};
\node[draw=none] (4a) at (2,-0.8) {$1$: $0.1$};
\node[draw=none] (4b) at (2,-1.2) {$2$: $0.4$};
\draw[blue] (0,1) ellipse (3cm and 1cm);
\draw[red, thick, dashed] (-2.7,1.7)--(2.7,1.7)--(2.7,-1.7)--(1.5,-1.7)--(-2.7,0.5)--cycle;
\path[draw,thick]
(1) edge node {} (2)
(2) edge node {} (4)
(3) edge node {} (2)
(3) edge node {} (4);
\end{tikzpicture}
\caption{An instance with $n=2$ agents and $m=4$ items.
The numbers in each vertex indicate the agents' utilities for the item.
The items assigned to agent~$2$ in the egalitarian-optimal (resp., utilitarian-optimal) connected allocation are those in the blue ellipse (resp., red dashed shape).}
\label{fig:example}
\end{figure}

We focus on two well-studied notions of social welfare: \emph{egalitarian welfare}, defined as the minimum among the agents' utilities, and \emph{utilitarian welfare}, defined as the sum of the agents' utilities.
For example, consider the instance in \Cref{fig:example}.
In this instance, the optimal (unconstrained) egalitarian and utilitarian welfare is $0.8$ and $1.7$, respectively.
Indeed, the optimal utilitarian welfare is attained by simply giving each item to the agent who values it higher.
In order to achieve an egalitarian welfare of at least $0.8$, each agent must receive the two items that she values at least $0.4$, leading to the same allocation.
However, the optimal egalitarian welfare subject to connectivity is $0.5$, with agent~$2$'s bundle corresponding to the blue ellipse.
Indeed, to obtain a higher egalitarian welfare, agent~$1$ would need the top left item together with at least one bottom item, but then she would also need the top right item due to the connectivity requirement, leaving agent~$2$ with utility at most $0.4$.
The optimal utilitarian welfare subject to connectivity is $1.4$, with agent~$2$'s bundle corresponding to the red dashed shape---to see this, note that the allocation that yields the optimal (unconstrained) utilitarian welfare is not connected, and assigning any item to an agent with lower value leads to a loss in utilitarian welfare of at least $0.3$, so the remaining utilitarian welfare is at most $1.7 - 0.3 = 1.4$.
Hence, due to the connectivity constraint, the optimal egalitarian welfare decreases by a factor of $0.8/0.5 = 1.6$, and the optimal utilitarian welfare decreases by a factor of $1.7/1.4 \approx 1.21$.

To quantify the potential welfare loss, we introduce the concept of \emph{egalitarian (resp., utilitarian) price of connectivity (PoC)} of a graph, defined as the worst-case ratio between the optimal egalitarian (resp., utilitarian) welfare of an arbitrary allocation and the optimal egalitarian (resp., utilitarian) welfare of a connected allocation, across all possible utility profiles of the agents.
By definition, a complete graph has a PoC of exactly~$1$ with respect to both welfare measures, and intuitively, the less connected a graph is, the higher its PoC becomes.
Determining the PoC of a graph could help central decision-makers assess whether the benefits of having connected allocations outweigh the possible loss in welfare.
Moreover, from a mathematical perspective, the PoC is an inherently interesting property of each graph.

\subsection{Our Results}

\renewcommand{\arraystretch}{1.4}

\begin{table}[t]
\centering
\begin{tabular}{@{}c|l|l@{}}
\hline
& \multicolumn{1}{c|}{Egalitarian PoC} & \multicolumn{1}{c}{Utilitarian PoC} \\
\hline \hline
\multirow{2}{*}{
\begin{tabular}{@{}l@{}}
$K_m$ with a non-empty \\ matching removed
\end{tabular}} & $\rule{0pt}{4.5ex}\begin{cases}
2 & \text{for}~L_3~\text{or}~L_5 \\
\frac{m-2}{m-3} & \text{otherwise}
\end{cases}$ & $\begin{cases}
\frac{4}{3} & \text{for}~L_3~\text{or}~L_5 \\
\frac{2m - 4}{2m - 5} & \text{otherwise}
\end{cases}$ \hfill (\Cref{thm:util:complete}) \\
& \hfill (\Cref{thm:egal:complete}) & \hfill  \\

\hline

$K_{|V_1|, |V_2|}$ ($2\le |V_1|\le |V_2|$) & $\rule{0pt}{4.5ex}\frac{|V_1|}{|V_1| - 1}$  \hfill (\Cref{thm:egal:bipartite})
& $\rule{0pt}{5.5ex}\begin{cases}
\frac{4}{3} & \text{if}~|V_1| = 2 \\
\frac{2}{2 - \frac{1}{|V_1|} - \frac{1}{|V_2|}}  & \text{otherwise}
\end{cases}
\rule[-4.5ex]{0pt}{0pt}$ (\Cref{thm:util:2-agent:complete-bipartite}) \\

\hline

Cycles & $2$ \hfill (\Cref{thm:egal:2-agent:2connected}) & $\rule{0pt}{3ex}\frac{2k}{k + 1}\rule[-1.5ex]{0pt}{0pt}$, where $k = \floor*{\frac{m}{2}}$ \hfill (\Cref{thm:util:2-agent:cycle}) \\

\hline

Trees & max.~degree \hfill (\Cref{thm:egal:2-agent-1connected}) & $\rule{0pt}{3ex}\frac{2 (m-1)}{m}\rule[-1.5ex]{0pt}{0pt}$ \hfill (\Cref{thm:util:2-agent:tree}) \\
\hline
\end{tabular}
\caption{Summary of our PoC results for $n = 2$ agents, where $m$ denotes the number of items, $K_r$ denotes the complete graph with~$r$ vertices, $K_{r,s}$ denotes the complete bipartite graph with $r$ and $s$ vertices on the two sides, $L_3$ denotes the graph $K_3$ with one edge removed, and $L_5$ denotes the graph $K_5$ with two disjoint edges removed.
\Cref{thm:egal:2-agent:2connected,thm:egal:2-agent-1connected} also apply to the larger classes of graphs with connectivity $2$ and $1$, respectively.}
\label{table:2-agent}
\end{table}

\begin{table}[t]
\centering
\begin{tabular}{@{}c|ll|ll@{}}
\hline
& \multicolumn{2}{c|}{Egalitarian PoC} & \multicolumn{2}{c}{Utilitarian PoC} \\
\hline \hline
Stars & $m-n+1$ & (\Cref{thm:egal:n-agent-star}) & $\Omega(n)$\ & (\Cref{thm:utilstarn}) \\
\hline
Paths & $\rule{0pt}{4.5ex}
\begin{cases}
m-n+1 & \text{if } n\leq m < 2n-1 \\
n & \text{if } 2n-1\leq m
\end{cases}
\rule[-3.5ex]{0pt}{0pt}$ & (\Cref{thm:egal:n-agent-path}) & $\Omega(n)$ & (\Cref{thm: pathuswlb}) \\
\hline
Cycles & $\rule{0pt}{6ex}
\begin{cases}
m-n+1 & \text{if } n\leq m < 2n-2 \\
n-1 & \text{if } 2n-2\leq m<n^2\\
n & \text{if } n^2\leq m
\end{cases}
\rule[-5ex]{0pt}{0pt}$ & (\Cref{thm:egal:n-agent-cycle}) & $\Omega(n)$ & (\Cref{thm:uswanyncycles}) \\
\hline
Any graph & $\leq m-n+1$ & (\Cref{lem:egal:n-agent-any-upper}) & $\leq n$ & (\Cref{prop:uswupperbound}) \\
\hline
\end{tabular}
\caption{Summary of our PoC results for any number of agents, where $n$ and $m$ denote the number of agents and items, respectively.}
\label{table:n-agent}
\end{table}

We derive tight or asymptotically tight bounds on the PoC for various classes of graphs.
Denote by $n$ and $m$ the number of agents and items, respectively.

In \Cref{sec:egal}, we investigate the egalitarian PoC, starting from the case of two agents and proceeding from dense graph classes to sparser ones.
For complete graphs with a non-empty matching removed, we show that with the exception of the $5$-item case, the PoC is $\frac{m-2}{m-3}$ regardless of how many edges are included in the matching.
For complete bipartite graphs, the PoC depends only on the number of vertices on the smaller side.
We then address graphs with vertex connectivity~$1$ or $2$, where the vertex connectivity of a graph refers to the smallest number~$k$ for which there exist $k$~vertices whose removal makes the graph disconnected.
For graphs with connectivity~$2$ (including cycles), the PoC is always~$2$, whereas for graphs with connectivity~$1$ (including trees), the PoC depends on the maximum number of components when a vertex is removed from the graph and, if this number is~$2$, on whether the ``block decomposition'' of the graph is a path.
In addition, we determine the PoC for trees in the case of three agents, as well as for stars, paths, and cycles when the number of agents can be arbitrary.

In \Cref{sec:util}, we turn our attention to the utilitarian PoC, again starting with two agents.
For complete graphs with a non-empty matching removed, we present a similar finding as for the egalitarian case: the PoC is independent of the number of edges in the matching.
On the other hand, we uncover differences between the two variants of PoC for the classes of complete bipartite graphs, trees, and cycles.
In particular, the utilitarian PoC for complete bipartite graphs depends on the sizes of both sides (unless one side has size~$2$, in which case the PoC is constant), while the utilitarian prices for trees and cycles depend only on the number of vertices in the graph.
Moreover, for any number of agents, we show that the utilitarian PoC is always at most $n$, and it is $\Omega(n)$ for stars, paths, and cycles when the number of vertices is large compared to~$n$.

Our results for two agents (resp., any number of agents) are summarized in Table~\ref{table:2-agent} (resp., Table~\ref{table:n-agent}). 

\subsection{Related Work}
\label{sec:related-work}

Our work is closely related to the literature of \emph{fair division}, which addresses fairness in resource allocation scenarios \citep{BramsTa96,RobertsonWe98,Moulin03}.
In that literature, the \emph{price of fairness} was introduced to quantify the worst-case welfare loss due to fairness considerations \citep{BertsimasFaTr11,CaragiannisKaKa12}, and commonly studied with respect to both utilitarian and egalitarian welfare \citep{AumannDo15,Suksompong19,BarmanBhSh20,BeiLuMa21,CelineDzKo23,LiLiLu24}.
By contrast, our work does not focus on fairness itself, but rather explores the trade-off between social welfare and connectivity constraints.
In a similar spirit, \citet{LamLiSun25} studied the loss in utilitarian and egalitarian welfare due to cardinality constraints.

Many researchers have investigated connectivity in resource allocation \citep{BouveretCeEl17,GoldbergHoSu20,GrecoSc20,LoncTr20,BiloCaFl22,CaragiannisMiSh22,Igarashi23,Lonc23,SunLi23,YuenSu24}.
Like us, several of them considered classes of graphs such as paths, stars, trees, and cycles.
Nevertheless, these authors have mostly focused on the fairness perspective, with the goal of satisfying fairness notions such as \emph{envy-freeness}.
An exception is the work by \citet{IgarashiPe19}, who examined the problem of finding a connected allocation that is \emph{Pareto-optimal}, meaning that no other connected allocation makes at least one of the agents better off and none of them worse off.
Pareto optimality can be viewed as a qualitative measure of economic efficiency, differing from the quantitative measures that we study.

The term \emph{price of connectivity (PoC)} was first used in the work of \citet{BeiIgLu22} to capture the price in terms of a fairness notion called the \emph{maximin share}.
Specifically, these authors defined the PoC as the worst-case ratio between the maximin share taken over all possible partitions of the items (into~$n$ parts) and that over all \emph{connected} partitions.
When translated to our setting, their maximin share PoC corresponds to our egalitarian PoC when agents have identical valuations.
Thus, their PoC can be lower than our PoC---for example, their PoC for graphs with connectivity~$2$ is $4/3$, whereas our egalitarian PoC for this graph class is $2$.
Another notable difference between the two versions of PoC is that for graphs with connectivity~$1$, the egalitarian PoC sometimes depends on the block decomposition, while the maximin share PoC does not.
We remark that \citet{BeiIgLu22} did not study any notions related to the utilitarian PoC.

\section{Preliminaries}

For any positive integer~$t$, let $[t] \coloneqq \{1, 2, \dots, t\}$.
Denote by~$N = [n]$ the set of~$n$ agents and~$M$ the set of~$m$ indivisible items.
There is a bijection between the items in~$M$ and the~$m$ vertices of a connected undirected graph~$G = (V,E)$; we will refer to items and vertices interchangeably.
Each agent~$i \in N$ has a non-negative utility~$u_i(\{g\})$ for each item~$g \in M$; for convenience, we sometimes write $u_i(g)$ instead of $u_i(\{g\})$.
We assume that utilities are \emph{additive}, i.e., $u_i(M') = \sum_{g \in M'} u_i(g)$ for all~$i \in N$ and~$M' \subseteq M$, as well as \emph{normalized}, that is, $u_i(M) = 1$ for all~$i \in N$;\footnote{The additivity assumption is often made in the fair division literature, and the normalization assumption is typical in research on the price of fairness \citep{CaragiannisKaKa12,AumannDo15,BeiLuMa21,CelineDzKo23}.}
we sometimes write $u_i(G)$ instead of $u_i(M)$.
Denote by~$\utilityProfile = (u_1, u_2, \dots, u_n)$ the \emph{utility profile} of the agents.
We refer to a setting with the set of agents~$N$, the set of items~$M$ and their underlying graph~$G$, and the utility profile~$\utilityProfile$ as an \emph{instance}, denoted by $I = \langle N, G, \utilityProfile \rangle$.

A \emph{bundle} is a (possibly empty) subset of items; it is called \emph{connected} if the items in the bundle form a connected subgraph of~$G$.
An \emph{allocation}~$\alloc = (M_1, M_2, \dots, M_n)$ is a partition of the items in~$M$ into $n$ bundles such that agent~$i \in N$ receives bundle~$M_i$.
Moreover, an allocation or a partition is \emph{connected} if all of its bundles are connected.
Denote by~$\connectedAllocsOf{I}$ the set of all connected allocations for instance~$I$.

In this paper, we will quantify the worst-case egalitarian and utilitarian welfare loss incurred when imposing the requirement that each agent must receive a connected bundle.
Given an instance~$I$ and an allocation~$\alloc = (M_1, M_2, \dots, M_n)$ of the instance,
\begin{itemize}
\item the \emph{egalitarian welfare} of~$\alloc$, denoted by~$\egalSW(\alloc)$, is the minimum among the agents' utilities, i.e., $\egalSW(\alloc) \coloneqq \min_{i \in N} u_i(M_i)$;

\item the \emph{utilitarian welfare} of~$\alloc$, denoted by~$\utilSW(\alloc)$, is the sum of the agents' utilities, i.e., $\utilSW(\alloc) \coloneqq \sum_{i \in N} u_i(M_i)$.
\end{itemize}

The optimal egalitarian welfare of instance~$I$, denoted by~$\egalOPT(I)$, is the maximum egalitarian welfare over all possible allocations in $I$.
The optimal utilitarian welfare and $\utilOPT(I)$ are defined analogously.

We now proceed to define the central concept of the paper---the \emph{price of connectivity (PoC)}---which captures the largest multiplicative gap between the optimal welfare among all allocations and that among all \emph{connected} allocations.

\begin{definition}[\egalPrice]
Given a graph~$G$ and a number of agents~$n$, the \emph{egalitarian price of connectivity (egalitarian PoC) of~$G$ for $n$ agents} is defined as\footnote{We interpret $\frac{0}{0}$ in this context to be equal to~$1$.
Note that $\egalOPT(I) = 0$ if and only if $\max_{\alloc \in \connectedAllocsOf{I}} \egalSW(\alloc) = 0$, because~$G$ is assumed to be connected, and both conditions are equivalent to the condition that there does not exist an allocation where each agent attains positive utility.
}
\[
\egalPrice(G, n) = \sup_{I = \langle N, G, \utilityProfile \rangle} \frac{\egalOPT(I)}{\max_{\alloc \in \connectedAllocsOf{I}} \egalSW(\alloc)},
\]
where the supremum is taken over all possible instances with the graph $G$ and $n$ agents.
For an instance~$I$, we call the ratio $\frac{\egalOPT(I)}{\max_{\alloc \in \connectedAllocsOf{I}} \egalSW(\alloc)}$ its \emph{egalitarian welfare ratio}.
\end{definition}

The utilitarian price of connectivity is defined in a similar manner.

\begin{definition}[\utilPrice]
Given a graph~$G$ and a number of agents~$n$, the \emph{utilitarian price of connectivity (utilitarian PoC) of~$G$ for $n$ agents} is defined as\footnote{Similarly to egalitarian welfare, $\utilOPT(I) = 0$ if and only if $\max_{\alloc \in \connectedAllocsOf{I}} \utilSW(\alloc) = 0$, as both conditions are equivalent to the condition that every agent has zero utility for every item.
}
\[
\utilPrice(G, n) = \sup_{I = \langle N, G, \utilityProfile \rangle} \frac{\utilOPT(I)}{\max_{\alloc \in \connectedAllocsOf{I}} \utilSW(\alloc)},
\]
where the supremum is taken over all possible instances with the graph $G$ and $n$ agents.
For an instance~$I$, we call the ratio $\frac{\utilOPT(I)}{\max_{\alloc \in \connectedAllocsOf{I}} \utilSW(\alloc)}$ its \emph{utilitarian welfare ratio}.
\end{definition}

Recall that a \emph{tree} is a connected graph without cycles.
A \emph{rooted tree} is a tree in which a specific vertex has been designated as the \emph{root vertex}. 
Given a rooted tree $T$ with root vertex $v$, a \emph{top-level subtree} $T'$ refers to a rooted tree with a root vertex $v'$ adjacent to $v$ such that a vertex belongs to $T'$ if and only if the path from it to $v'$ does not contain $v$.

\section{Egalitarian Price of Connectivity}
\label{sec:egal}

In this section, we investigate the loss of egalitarian welfare due to connectivity constraints.
We remark that given any graph~$G$, if $m < n$, then the optimal egalitarian welfare is always~$0$ and thus $\egalPrice(G, n) = 1$.
We therefore assume that $m \geq n$ for the rest of this section.

We begin by establishing a general upper bound on the egalitarian PoC for any graph~$G$.
The intuition behind the bound is that for any instance~$I$, by giving every agent her most preferred item from an egalitarian-optimal allocation, each agent receives utility at least~$\frac{\egalOPT(I)}{m-n+1}$.
This is because each bundle in an egalitarian-optimal allocation contains at most~$m - n + 1$ items.

\begin{proposition}
\label{lem:egal:n-agent-any-upper}
For any graph~$G$, it holds that $\egalPrice(G, n) \leq m - n + 1$.
\end{proposition}

\begin{proof}
Given any instance~$I$, let~$\alloc = (M_1, M_2, \dots, M_n)$ be an egalitarian-optimal allocation in the instance.
If $\egalOPT(I) = \min_{i \in N} u_i(M_i) = 0$, then $\egalPrice(G, n) = 1$.
We thus assume that $\min_{i \in N} u_i(M_i) > 0$, which implies that  $1 \leq |M_i| \leq m - n + 1$ for all~$i \in N$.
By giving each agent~$i \in N$ her most preferred item in~$M_i$, such a (partial) allocation gives the agent a utility of at least
\[
\frac{u_i(M_i)}{|M_i|} \geq \frac{\egalOPT(I)}{|M_i|} \geq \frac{\egalOPT(I)}{m-n+1}.
\]
We then extend the partial allocation to a complete, connected allocation in an arbitrary way; this is possible since $G$ is connected.
It follows that $\egalPrice(G, n) \leq m - n + 1$.
\end{proof}

We next introduce a reduction which will be used throughout this section to simplify our proofs for the egalitarian PoC upper bounds across various graphs.
Given a graph~$G$, denote by~$\mathcal{I}_G$ the set of instances such that each vertex in~$G$ is positively valued by at most one agent.
Note that, unlike in the rest of the paper, we do \emph{not} assume that the agents' utilities are normalized for instances in~$\mathcal{I}_G$.
This is crucial because the instance $\widehat{I}$ constructed in the proof of the lemma will not be normalized.

\begin{lemma}
\label{prop:structured-instance}
Let~$\beta \in (0, 1]$, and fix any graph~$G$.
Suppose that for every instance~$\widehat{I} \in \mathcal{I}_G$, there exists a connected allocation~$\widehat{\alloc}$ such that $\egalSW(\widehat{\alloc}) \geq \beta \cdot \egalOPT(\widehat{I})$.
Then, $\egalPrice(G, n) \leq \frac{1}{\beta}.$
\end{lemma}

\begin{proof}
It suffices to show that given any instance $I = \langle N, G, \utilityProfile \rangle$, there always exists a connected allocation with egalitarian welfare at least $\beta \cdot \egalOPT(I)$.

Denote by~$(M_1, M_2, \dots, M_n)$ an egalitarian-optimal allocation in instance~$I$.
Construct an instance~$\widehat{I}$ with the same graph~$G$ such that for each~$i \in N$, agent~$i$ values all vertices in $M_i$ the same as in $I$ but values all other vertices at~$0$.
It is clear that $\widehat{I} \in \mathcal{I}_G$.
As assumed in the lemma statement, there exists a connected allocation~$\widehat{\alloc}$ in $\widehat{I}$ such that $\egalSW(\widehat{\alloc}) \geq \beta \cdot \egalOPT(\widehat{I})$.

Observe that $I$ and $\widehat{I}$ have the same optimal egalitarian welfare, i.e., $\egalOPT(I) = \egalOPT(\widehat{I})$.
Moreover, the egalitarian welfare of $\widehat{\alloc}$ in $I$ is at least that in $\widehat{I}$.
Therefore, in instance~$I$, the connected allocation~$\widehat{\alloc}$ gives an egalitarian welfare of at least $\beta \cdot \egalOPT(\widehat{I}) = \beta \cdot \egalOPT(I)$, as desired.
\end{proof}

\Cref{prop:structured-instance} implies that when proving upper bounds on the egalitarian PoC, it suffices to focus only on the (not necessarily normalized) instances where each vertex is valued by at most one agent.
Specifically, if the egalitarian PoC across all such instances is at most $1/\beta$ for some $\beta\in(0,1]$, then so is the actual egalitarian PoC (across all normalized instances where each vertex can be valued by multiple agents).
In the remainder of this section, we first address the cases of two or three agents,\footnote{We remark that several applications of resource allocation involve a small number of agents. 
For example, see Section~1.1.1 in the work by \citet{PlautRo20} for a discussion on the importance of the two-agent case.} before proceeding to the general case.

\subsection{Two Agents}
\label{sec:egal:2agent}

For two agents, we begin by considering classes of dense graphs.
While a complete graph has a PoC of $1$ by definition, we find that as soon as we remove an edge from it, the egalitarian PoC becomes strictly larger than $1$, i.e., there exists an instance such that the optimal egalitarian welfare overall cannot be attained by a connected allocation.
On the other hand, with the exception of the $5$-item case, the egalitarian PoC does not increase further when additional disjoint edges are removed from the graph.
The reason why the graph $L_5$ (\Cref{fig:L_5}) constitutes an exceptional case is that, unlike other graphs in this class with at least $5$~vertices, there are two pairs of vertices in $L_5$---i.e., $(v_{1,1},v_{1,2})$ and $(v_{2,1},v_{2,2})$---such that there do not exist two disjoint paths with one path connecting each pair.\footnote{In other words, $L_5$ is not ``$2$-linked''---see the discussion before \Cref{lem:egal:2linked}.}
Recall that $K_m$ denotes the complete graph with~$m$ vertices.

\begin{theorem}
\label{thm:egal:complete}
Let $G$ be a complete graph with a non-empty matching removed.
Then, $\egalPrice(G, 2) = 2$ if $G$ is
\begin{itemize}
\item $L_3$, which is $K_3$ with an edge removed, or
\item $L_5$, which is $K_5$ with two disjoint edges removed,
\end{itemize}
and $\egalPrice(G, 2) = \frac{m-2}{m-3}$ otherwise.
\end{theorem}

\begin{proof}
We first handle the case $m\le 4$.
When $m = 3$, the egalitarian PoC of~$2$ for~$L_3$ follows from \Cref{thm:egal:n-agent-path} for paths.
When $m = 4$, the egalitarian PoC of~$2 = \frac{4-2}{4-3}$ for~$K_4$ with two disjoint edges removed follows from \Cref{thm:egal:n-agent-cycle} for cycles.
Since the PoC cannot increase when an edge is added to a graph, this implies that for $K_4$ with a single edge removed, the egalitarian PoC is at most~$2$.
We now provide an instance showing that the egalitarian PoC for this graph is at least~$2$, and therefore exactly~$2$.
Consider an instance~$I$ where agent~$1$ values the two vertices without an edge between them at~$1/2$ each, and agent~$2$ values the two remaining vertices at~$1/2$ each.
We have $\egalOPT(I) = 1$ and $\max_{\alloc \in \connectedAllocsOf{I}} \egalSW(\alloc) = 1/2$, so the egalitarian PoC is at least~$2$.

For the remainder of the proof, let $m \geq 5$.
We establish the lower bound and the upper bound in turn.

\begin{figure}[t]
\centering
\begin{tikzpicture}
\draw (2,4) -- (0.1,2.62) -- (3.9,2.62) -- (0.82,0.38) -- (3.18,0.38) -- (2,4) -- (0.82,0.38);
\draw (3.18,0.38) -- (0.1,2.62);
\draw (2,4) -- (3.9,2.62);
\draw[fill=white] (2,4) node (v5) {$v_5$} circle [radius = 0.35];
\draw[fill=white] (0.1,2.62) node[label=left:{$1$: $0.5$}] (v11) {$v_{1, 1}$} circle [radius = 0.35];
\draw[fill=white] (0.82,0.38) node[label=left:{$1$: $0.5$}] (v12) {$v_{1, 2}$} circle [radius = 0.35];
\draw[fill=white] (3.9,2.62) node[label=right:{$2$: $0.5$}] (v21) {$v_{2, 1}$} circle [radius = 0.35];
\draw[fill=white] (3.18,0.38) node[label=right:{$2$: $0.5$}] (v22) {$v_{2, 2}$} circle [radius = 0.35];
\end{tikzpicture}
\caption{The graph~$L_5$ and a utility profile showing that its egalitarian PoC is at least~$2$.
The number at each vertex indicates the corresponding agent's utility for the item.
All utilities not indicated are zero.}
\label{fig:L_5}
\end{figure}

\paragraph{Lower Bound (for $m \geq 5$):}
First, consider $L_5$, and denote the missing edges by $\{v_{1,1}, v_{1,2}\}$ and $\{v_{2,1}, v_{2,2}\}$ (see \Cref{fig:L_5}).
For~$i \in [2]$, let agent~$i$ value vertices $v_{i,1}$ and $v_{i,2}$ at~$1/2$ each, and let the final vertex~$v_5$ be non-valued.
Observe that either agent can obtain both of her valued vertices by taking~$v_5$, but if one agent does so, then the other agent can only obtain one of her valued vertices.
Hence, the optimal egalitarian welfare under connectivity is at most~$1/2$, and the egalitarian PoC is at least~$2$.

Consider now any graph that is not~$L_5$.
Let~$v_1$ and~$v_2$ be two vertices without an edge between them, let agent~$1$ value each of these vertices at $1/2$, and let agent~$2$ value each of the remaining $m-2$ vertices at $\frac{1}{m-2}$.
An egalitarian-optimal connected allocation assigns to agent~$1$ both of her valued vertices along with one of agent~$2$'s valued vertices, and assigns to agent~$2$ all but one of her valued vertices.
This leads to an egalitarian welfare of $\frac{m-3}{m-2}$, so the egalitarian PoC is at least $\frac{m-2}{m-3}$.

\paragraph{Upper Bound (for $m \geq 5$):}
To prove the upper bound, we show that among all possible instances, those described in the lower bound examples maximize the egalitarian welfare ratio $\frac{\egalOPT(I)}{\max_{\alloc \in \connectedAllocsOf{I}} \egalSW(\alloc)}$.
By \Cref{prop:structured-instance}, it suffices to consider instances~$I$ in which each vertex is positively valued by at most one agent; note that $I$ is not necessarily normalized.
Denote the missing edges by $\{v_{1,1}, v_{1,2}\}$, $\{v_{2,1}, v_{2,2}\}$, $\dots$, $\{v_{k,1}, v_{k,2}\}$.
If there exists $i \in [k]$ such that neither agent positively values both $v_{i,1}$ and $v_{i,2}$ at the same time, then both agents can simultaneously receive all of their valued vertices in a connected allocation, as $v_{i,1}$ is connected to every vertex except $v_{i,2}$, and similarly $v_{i,2}$ is connected to every vertex except $v_{i,1}$.
For such an instance~$I$, it holds that $\max_{\alloc \in \connectedAllocsOf{I}} \egalSW(\alloc) = \egalOPT(I)$.
Therefore, to maximize the egalitarian welfare ratio, we may assume that for each $i\in [k]$, both $v_{i,1}$ and $v_{i,2}$ are positively valued by the same agent; without loss of generality, let agent~$1$ positively value $v_{1, 1}$ and $v_{1, 2}$.
Moreover, we may assume that each agent positively values at least two vertices, since otherwise both agents can simultaneously receive all of their valued vertices in a connected allocation.

We first consider the special case where the graph is~$L_5$.
From the observations above, we may assume that agent~$1$ positively values $v_{1,1}$, $v_{1,2}$, and $v_5$, while agent~$2$ positively values $v_{2,1}$ and $v_{2,2}$.
Suppose without loss of generality that $u_1(v_{1,1}) \le u_1(v_{1,2})$.
Then, agent~$1$ can obtain utility at least $u_1(G)/2$ by taking $v_{1,2}$ and $v_5$, while agent~$2$ can receive both of her valued items by taking $v_{1,1}$, $v_{2,1}$, and $v_{2,2}$; note that the resulting allocation is connected.
Hence, the egalitarian PoC is at most $2$.

Next, we consider the general case where the graph is not~$L_5$.
Note that any subset of at least three vertices form a connected subgraph of the graph.
From the observations above, we may assume that agent~$1$ positively values $v_{1,1}$ and $v_{1,2}$, and agent~$2$ positively values at least two vertices.
We claim that if some vertex is non-valued by both agents, then both agents can receive all of their valued vertices in a connected allocation.
To see this, suppose that some vertex is non-valued by both agents; this vertex must be different from $v_{1,1}$ and $v_{1,2}$.
If the graph is~$K_5$ with a single edge removed, then agent~$2$'s valued vertices must form a connected subgraph, so agent~$1$ can take the non-valued vertex to obtain all of her valued vertices.
If $m \geq 6$ and agent~$1$ values some vertex other than $v_{1,1}$ and $v_{1,2}$, then agent~$1$'s valued vertices form a connected subgraph, and agent~$2$ can take the non-valued vertex to obtain all of her valued vertices.
If $m \geq 6$ and agent~$1$ only values $v_{1,1}$ and $v_{1,2}$, then either agent~$2$ values at least three vertices, or there are at least two non-valued vertices.
In either case, both agents can receive all of their valued vertices in a connected allocation, with agent~$1$ taking a non-valued vertex along with all of her valued vertices, and agent~$2$ taking the other non-valued vertex if one exists, or otherwise all of her valued vertices.

It remains to consider the case where every vertex is positively valued by some agent.
If both agents value at least three vertices each, then each agent's valued vertices form a connected subgraph.
Hence, assume without loss of generality that agent~$1$ values only two vertices, $v_{1,1}$ and $v_{1,2}$, and agent~$2$ values all other vertices.
Consider the allocation~$\alloc$ that assigns $v_{1,1}$, $v_{1,2}$, and agent~$2$'s least-valued vertex to agent~$1$, and the remaining vertices to agent~$2$; note that $\alloc$ is connected.
In $\alloc$, agent~$1$ receives utility $u_1(G)$, while agent~$2$ receives utility at least $\frac{m-3}{m-2}\cdot u_2(G)$.
Since $\egalOPT(I) \le \min\{u_1(G), u_2(G)\}$, the egalitarian welfare of~$\alloc$ is at least $\frac{m-3}{m-2}\cdot \egalOPT(I)$.
It follows that the egalitarian PoC is at most $\frac{m-2}{m-3}$, as desired.
\end{proof}

We remark that a complete graph with a matching removed can be viewed as a \emph{complete $k$-partite graph} $K_{n_1, n_2, \dots, n_k}$ where $n_1, \dots ,n_k \in \{1, 2\}$---the vertices are partitioned into~$k$ independent sets of size~$1$ or~$2$, and there is an edge between every pair of vertices from different independent sets~\citep[p.~41]{ChartrandZh20}.
This class forms a subclass of \emph{Tur\'{a}n graphs}~\citep[p.~108]{Bollobas98}, but is more general than \emph{hyperoctahedral graphs}~\citep[p.~17]{Biggs93}.

We now address another class of dense graphs: complete bipartite graphs.
For this class, we find that the egalitarian PoC depends only on the number of vertices on the smaller side of the graph (unless that number is~$1$).

\begin{theorem}
\label{thm:egal:bipartite}
Let $G$ be a complete bipartite graph with $x$ vertices on one side and at least $x$ vertices on the other side.
Then,
\[
\egalPrice(G, 2) = \begin{cases}
m - 1 & \text{ if } x = 1; \\
\frac{x}{x-1} & \text{ otherwise}.
\end{cases}
\]
\end{theorem}

\begin{proof}
The case $x = 1$ follows from \Cref{thm:egal:n-agent-star} for stars, so we assume that $x \geq 2$.
Denote the two vertex sets by~$V_1$ and~$V_2$.
For the lower bound, consider an instance~$I$ where agent~$1$ values each vertex in~$V_1$ at~$1/|V_1|$ and agent~$2$ values each vertex in~$V_2$ at~$1/|V_2|$.
We have $\egalOPT(I) = 1$ and $\max_{\alloc \in \connectedAllocsOf{I}} \egalSW(\alloc) = \frac{x-1}{x}$, so the egalitarian PoC is at least $\frac{x}{x-1}$.

It remains to prove the upper bound.
By \Cref{prop:structured-instance}, it suffices to consider instances in which each vertex is positively valued by at most one agent.
Since~$G$ is a complete bipartite graph, an agent that receives a vertex on one side can also receive any subset of vertices on the other side in a connected allocation.
If both agents have valued vertices on both sides, we can obtain the optimal egalitarian welfare overall via a connected allocation.
Hence, assume without loss of generality that all of agent~$1$'s valued vertices are in $V_1$.

If agent~$2$ values some vertex in $V_1$, we let agent~$2$ take all of her valued vertices in $V_1$ along with all vertices in $V_2$ except her least valued one (possibly of value~$0$), and let agent~$1$ take the remaining vertices; note that the resulting allocation is connected.
Agent~$1$ receives all of her valued vertices, while agent~$2$ receives utility at least $\frac{|V_2|-1}{|V_2|}\cdot u_2(G) \ge \frac{x-1}{x}\cdot u_2(G)$.
Finally, suppose that agent~$2$ does not value any vertex in $V_1$, so all of her valued vertices are in $V_2$.
In this case, we let agent~$1$ take all vertices in $V_1$ except her least valued one (possibly of value~$0$), and give this vertex to agent~$2$.
Similarly, we let agent~$2$ take all vertices in $V_2$ except her least valued one, and give this vertex to agent~$1$.
Note that the resulting allocation is connected, and each agent $i\in [2]$ receives utility at least $\frac{|V_i|-1}{|V_i|}\cdot u_i(G) \ge \frac{x-1}{x}\cdot u_i(G)$.
It follows that the egalitarian PoC is at most $\frac{x}{x-1}$.
\end{proof}

Next, we present results for graphs classified by \emph{vertex connectivity}, sometimes referred to simply as \emph{connectivity}.
A graph is said to have connectivity~$k$ if there exist~$k$ vertices whose removal results in the graph being disconnected and $k$ is the smallest number with this property.

We first consider graphs with connectivity~$2$; we will show that all such graphs have an egalitarian PoC of~$2$.
To this end, we will use two lemmas, one for the upper bound and the other for the lower bound.
A \emph{bipolar ordering} of a graph is a one-to-one assignment of the integers $1,\dots,m$ to its vertices such that the vertex with each number $i\in \{2,3,\dots,m-1\}$ is adjacent to some vertex with a higher number and some vertex with a lower number.

\begin{lemma}
\label{lem:egal:bipolar}
Suppose that a graph~$G$ admits a bipolar ordering.
Then, $\egalPrice(G, 2) \leq 2$.
\end{lemma}

\begin{proof}
By \Cref{prop:structured-instance}, it suffices to consider instances in which each item is positively valued by at most one agent. 
Take a bipolar ordering, and consider the least $k$ such that for some agent~$i\in [2]$, the vertices $1,2,\dots,k$ together yield value at least $u_i(G)/2$.
Give these vertices to agent~$i$, and the remaining vertices to the other agent~$j$; by definition of bipolar ordering, the resulting allocation is connected.
Since each vertex is valued by at most one agent, agent $i$'s bundle must be worth less than $u_j(G)/2$ to agent $j$, so agent~$j$ receives utility at least $u_j(G)/2$.
Since both agents receive at least half of their utility for the entire set of vertices, the egalitarian PoC is at most~$2$.
\end{proof}

A graph is called \emph{$2$-linked} if for any disjoint pairs of vertices $(a,b)$ and $(c,d)$, there exist two vertex-disjoint paths, one from $a$ to $b$ and the other from $c$ to $d$.
Observe that every graph with connectivity~$2$ is not $2$-linked.
To see this, let $a$ and $b$ be vertices whose removal makes the graph disconnected, and let $c$ and $d$ be vertices belonging to different components of the resulting graph.
Then, any path between $c$ and $d$ must go through either $a$ or $b$.
Even a graph with connectivity~$3$ may not be $2$-linked, as demonstrated by the graph~$L_5$ in \Cref{fig:L_5} with $(a,b,c,d) = (v_{1,1},v_{1,2},v_{2,1},v_{2,2})$.
We remark that $2$-linked graphs and their generalizations have received significant interest from graph theory researchers \citep{Jung70,LarmanMa70,Thomassen80}.

\begin{lemma}
\label{lem:egal:2linked}
Suppose that a graph~$G$ is not $2$-linked. Then, $\egalPrice(G, 2) \geq 2$.
\end{lemma}

\begin{proof}
Let $(a,b)$ and $(c,d)$ be disjoint pairs of vertices such that there do not exist two disjoint paths, one from $a$ to $b$ and the other from $c$ to $d$.
Let agent~$1$ value~$a$ and~$b$ at~$1/2$ each, and let agent~$2$ value~$c$ and~$d$ at~$1/2$ each.
If agent~$1$ receives utility higher than~$1/2$, her bundle must contain a path that connects $a$ to $b$.
However, this means that agent~$2$ cannot receive both $c$ and $d$ at the same time.
It follows that the optimal egalitarian welfare among connected allocations is at most $1/2$, and the egalitarian PoC is at least $2$.
\end{proof}

We are ready to establish the egalitarian PoC for graphs with connectivity~$2$.

\begin{theorem}
\label{thm:egal:2-agent:2connected}
Let $G$ be a graph with connectivity~$2$. Then, $\egalPrice(G, 2) = 2$.
\end{theorem}

\begin{proof}
Any graph with connectivity~$2$ has a bipolar ordering \citep{MaonScVi86}, so the upper bound is a direct consequence of \Cref{lem:egal:bipolar}.
For the lower bound, recall from the paragraph before \Cref{lem:egal:2linked} that a graph with connectivity~$2$ is not $2$-linked.
Therefore, the lower bound follows from \Cref{lem:egal:2linked}. 
\end{proof}

Next, we consider graphs with connectivity~$1$, which include trees. 
We find that the PoC depends on the maximum number of connected components when a single vertex is removed and, if this number is~$2$, on whether the ``block decomposition'' of the graph is a path.
(In particular, if the graph is a tree, this number is simply the maximum degree among the vertices.)
A \emph{block} is a maximal subgraph with connectivity at least~$2$, and a \emph{cut vertex} is a vertex whose removal disconnects a graph.
The \emph{block decomposition} of a graph~$G$ is a bipartite graph with all blocks of~$G$ on one side and all cut vertices of~$G$ on the other side; there is an edge between a block and a cut vertex in the bipartite graph if and only if the cut vertex belongs to the block in~$G$.
See \Cref{fig:bd} for an illustration of a block decomposition.
For any connected graph~$G$, the block decomposition of~$G$ is a tree~\citep[p.~121]{BondyMu08}.

\begin{theorem}
\label{thm:egal:2-agent-1connected}
Let $G$ be a graph with connectivity~$1$.
Denote by $d$ the maximum number of connected components when a vertex is removed from~$G$.
If $d=2$, then $\egalPrice(G, 2)=2$ when the block decomposition of~$G$ is a path, and $\egalPrice(G, 2)=3$ otherwise.
If $d\geq 3$, then $\egalPrice(G, 2)=d$.
\end{theorem}

\begin{proof}
We establish the lower bound and the upper bound in turn.

\paragraph{Lower Bound:}
For any $d\ge 2$, we can obtain a lower bound of $d$ as follows.
Consider a vertex whose removal results in $d$ connected components, let agent~$1$ value this vertex at~$1$, and let agent~$2$ value an arbitrary vertex in each of the $d$ components at $1/d$.
The optimal egalitarian welfare among connected allocations is $1/d$, so the egalitarian PoC is at least $d$.

\begin{figure}[t]
\centering
\begin{subfigure}{.45\linewidth}
\centering
\begin{tikzpicture}
\draw (-1.8,0) -- (1.8,0);
\draw (0,0.8) -- (0,2);
\draw[fill=white] (0,1.732/2) node (v1) {$v_1$} circle [radius=.3];
\draw[fill=white] (-1/2,0) node (v2) {$v_2$} circle [radius=.3];
\draw[fill=white] (1/2,0) node (v3) {$v_3$} circle [radius=.3];
\draw[fill=white] (0,2.166) node (v4) {} circle [radius=.3];
\draw[fill=white] (-1.8,0) node (v5) {} circle [radius=.3];
\draw[fill=white] (1.8,0) node (v6) {} circle [radius=.3];

\draw (v1) -- (v2) -- (v3) -- (v1);
\end{tikzpicture}
\caption{A graph with connectivity~$1$.}
\label{fig:bd:graph}
\end{subfigure}
\hfil
\begin{subfigure}{.45\linewidth}
\centering
\tikzstyle{vertex} = [circle,minimum size=5pt,inner sep=0pt,draw]
\begin{tikzpicture}
\node[vertex,fill] (cblock) {};
\node[vertex,fill=white,right=of cblock] (v3) {};
\node[vertex,fill,right=of v3] (v6) {};
\node[vertex,fill=white,left=of cblock] (v2) {};
\node[vertex,fill,left=of v2] (v5) {};
\node[vertex,fill=white,above=of cblock] (v1) {};
\node[vertex,fill,above=of v1] (v4) {};

\draw (v5) -- (v2) -- (cblock) -- (v3) -- (v6);
\draw (cblock) -- (v1) -- (v4);
\end{tikzpicture}
\caption{A corresponding block decomposition.}
\label{fig:bd:decomposition}
\end{subfigure}
\caption{An illustration for the proof of \Cref{thm:egal:2-agent-1connected}.
\Cref{fig:bd:graph} displays a graph such that the maximum number of connected components after removing a vertex is~$2$, i.e., $d = 2$; vertices~$v_1, v_2, v_3$ are cut vertices.
\Cref{fig:bd:decomposition} demonstrates its block decomposition, which is not a path; black vertices correspond to blocks and white vertices correspond to cut vertices.}
\label{fig:bd}
\end{figure}

It remains to show that the egalitarian PoC is at least~$3$ if $d=2$ and the block decomposition of~$G$ is not a path (see \Cref{fig:bd} for an example).
Since $d = 2$, the block decomposition consists of at least two blocks, and each cut vertex is adjacent to at most two blocks.
Recall that the block decomposition is a tree.
If all blocks have degree $1$ or $2$ in the block decomposition, then the block decomposition is a path.
Hence, there exists a block $B$ with degree at least $3$.
Let agent~$1$ value three cut vertices of~$B$ at $1/3$ each.
When these three cut vertices are removed, the resulting graph contains connected components $G_1,G_2,G_3$ disjoint from $B$.
Let agent~$2$ value a single vertex in each of $G_1,G_2,G_3$ at $1/3$.
Observe that agent~$2$ can obtain utility more than $1/3$ from a connected bundle only if the bundle includes at least two of agent~$1$'s cut vertices.
Hence, the optimal egalitarian welfare among connected allocations is $1/3$, and the egalitarian PoC is at least~$3$.

\begin{algorithm}[t]
\caption{For a graph with connectivity~$1$ and $n = 2$ agents}
\label{alg:egal:1-connected}
\DontPrintSemicolon

\KwIn{An instance~$I = \langle [2], G, \utilityProfile \rangle$.
(Each vertex in~$G$ is positively valued by at most one agent.)}

$\gamma \gets \max\{d, 3\}$, where $d$ is the maximum number of connected components resulting from the removal of a vertex.\;
Take any spanning tree~$T$ of~$G$, and root the tree at an arbitrary vertex. \;

\While{every top-level subtree of~$T$ is worth less than~$\frac{u_i(G)}{\gamma}$ for all~$i \in [2]$}{ \label{alg:egal:1-connected:firstmerge}
	Take two top-level subtrees of~$T$ that are connected by an edge in~$G$, merge them into one connected subtree with that edge, and remove the edge connecting one of the original subtrees to the root.\;
}

Take a top-level subtree~$T'$ of~$T$ such that $u_i(T') \geq \frac{u_i(G)}{\gamma}$ for some agent~$i$.\; \label{alg:egal:1-connected:first-subtree}

\While{there exists a top-level subtree $T''$ of~$T'$ worth at least $\frac{u_i(G)}{\gamma}$ to some agent~$i$}{\label{alg:egal:1-connected:further-subtree}
	$T'\leftarrow T''$\;
}

\eIf{there exists $i\in [2]$ such that $u_i(T') \geq \frac{u_i(G)}{\gamma}$ and $u_j(G \setminus T') \geq \frac{u_j(G)}{\gamma}$ for the other agent~$j$}{\Return{$T'$ for agent $i$ and $G\setminus T'$ for agent $j$} \label{alg:egal:1-connected:first-terminate}
}{
	Re-root the entire tree at the root of $T'$. \label{alg:egal:1-connected:root}
}

\While{every top-level subtree of $T$ is worth less than $\frac{u_i(G)}{\gamma}$ for all $i\in [2]$}{\label{alg:egal:1-connected:merge}
	Take two top-level subtrees of~$T$ that are connected by an edge in~$G$, merge them into one connected subtree with that edge, and remove the edge connecting one of the original subtrees to the root.\;
}

Take a top-level subtree $T'$ of $T$ such that $u_i(T')\geq \frac{u_i(G)}{\gamma}$ for some agent $i$.\;\label{alg:egal:1-connected:final}

\Return{$T'$ for agent $i$ and $G\setminus T'$ for the other agent $j$}
\end{algorithm}

\paragraph{Upper Bound:}
If $d=2$ and the block decomposition of~$G$ is a path, then~$G$ admits a bipolar ordering~\citep[Thm.~3.10]{BiloCaFl22}, so by \Cref{lem:egal:bipolar}, the egalitarian PoC is at most~$2$.
We handle the remaining cases by showing that $\egalPrice(G, 2) \leq \gamma$, where $\gamma \coloneqq \max\{d,3\}$.
By \Cref{prop:structured-instance}, it suffices to consider instances in which each vertex is positively valued by at most one agent.
We give an algorithm that produces a connected subgraph~$T'$ of~$G$ such that $T'$ is worth at least $\frac{u_i(G)}{\gamma}$ to some agent~$i$, and $G \setminus T'$ is also connected and worth at least $\frac{u_j(G)}{\gamma}$ to the other agent~$j$.
The algorithm is shown as \Cref{alg:egal:1-connected}.

We now demonstrate the correctness of the algorithm.
First, we show that the while-loop in \cref{alg:egal:1-connected:firstmerge} eventually terminates.
By the definition of~$d$, after all top-level subtrees are merged whenever possible, there will be at most~$d \le \gamma$ top-level subtrees remaining.
Since the root is positively valued by at most one agent, a top-level subtree that satisfies the condition in \cref{alg:egal:1-connected:first-subtree} exists, and the while-loop terminates.
For the same reason, the while-loop in \cref{alg:egal:1-connected:merge} eventually terminates as well.
If the algorithm terminates in \cref{alg:egal:1-connected:first-terminate}, then $T'$ and $G \setminus T'$ satisfy the utility requirements, and both of them are connected because $T$ is a spanning tree.

Suppose that the algorithm does not terminate in \cref{alg:egal:1-connected:first-terminate}.
If $u_i(T')\geq \frac{u_i(G)}{\gamma}$ for some agent~$i$ and $u_j(T')<\frac{u_j(G)}{\gamma}$ for the other agent~$j$, then $u_j(G \setminus T')>u_j(G)(1 - \frac{1}{\gamma}) \geq \frac{u_j(G)}{\gamma}$, and the algorithm would terminate in \cref{alg:egal:1-connected:first-terminate}, a contradiction.
Hence, $u_1(T')\geq \frac{u_1(G)}{\gamma}$ and $u_2(T')\geq \frac{u_2(G)}{\gamma}$.
This means that $u_1(G \setminus T')<\frac{u_1(G)}{\gamma}$ and $u_2(G \setminus T')<\frac{u_2(G)}{\gamma}$, as otherwise we would be able to allocate $G \setminus T'$ to one agent and $T'$ to the other agent in \cref{alg:egal:1-connected:first-terminate}.
Furthermore, by \cref{alg:egal:1-connected:further-subtree}, every top-level subtree of~$T'$ is worth less than $\frac{u_1(G)}{\gamma}$ to agent~$1$ and less than $\frac{u_2(G)}{\gamma}$ to agent~$2$.

When the tree is re-rooted in \cref{alg:egal:1-connected:root}, the previous $G \setminus T'$ becomes a top-level subtree of~$T$, and the other top-level subtrees of~$T$ are the same as the top-level subtrees of the previous $T'$; these subtrees are also worth less than $\frac{u_1(G)}{\gamma}$ to agent~$1$ and less than $\frac{u_2(G)}{\gamma}$ to agent~$2$.
As explained earlier, the while-loop in \cref{alg:egal:1-connected:merge} eventually terminates, so in \cref{alg:egal:1-connected:final}, $u_i(T')\geq \frac{u_i(G)}{\gamma}$ for some agent~$i$ and $u_j(T')<\frac{2u_j(G)}{\gamma}$ for the other agent~$j$.
Indeed, this is because the two top-level subtrees forming~$T'$ were worth less than $\frac{u_1(G)}{\gamma}$ to agent~$1$ and less than $\frac{u_2(G)}{\gamma}$ to agent~$2$ before merging.
It follows that $u_j(G \setminus T')> u_j(G)(1-\frac{2}{\gamma})\geq \frac{u_j(G)}{\gamma}$, so $T'$ and $G\setminus T'$ satisfy the utility requirements.
\end{proof}

While we have determined the egalitarian PoC for several classes of graphs, it remains an intriguing question to establish the PoC for every possible graph in the two-agent case.
An observation is that the PoC cannot be characterized solely in terms of the vertex connectivity.
Indeed, there exist graphs with connectivity~$3$ (\Cref{fig:L_5}) or even connectivity~$5$ \citep{Meszaros15} that are not $2$-linked---by \Cref{lem:egal:2linked}, the PoC of such graphs is at least $2$.
Moreover, since any such graph admits a bipolar ordering \citep{MaonScVi86}, combined with \Cref{lem:egal:bipolar}, the PoC of these graphs must be exactly~$2$.
On the other hand, the complete graph $K_m$ with a single edge removed has connectivity $m-2$ for each $m\ge 3$.
This means that $K_5$ has connectivity~$3$ and PoC $3/2$, whereas $K_7$ has connectivity~$5$ and PoC only $5/4$.
Hence, there exist graphs that have the same connectivity (e.g., $3$ or $5$) but different PoC.
Moreover, a graph with connectivity~$5$ can have the same PoC of~$2$ as all graphs with connectivity~$2$ due to \Cref{thm:egal:2-agent:2connected}.
In fact, even with identical valuations, the egalitarian PoC for two agents remains unknown---\citet[Conjecture~3.10]{BeiIgLu22} conjectured that the PoC can be characterized by linkedness in that case.
Settling their conjecture could be an important step toward a characterization in our setting where agents can have heterogeneous valuations.

\subsection{Three Agents}
\label{sec:egal:3agent}

Next, we consider the case of three agents.
For this case, we show that the egalitarian PoC of a tree is equal to the maximum number of connected components resulting from the removal of \emph{two} vertices; we denote this number by~$\delta$.

\begin{theorem}
\label{thm:egal:tree}
Let $G$ be a tree, and let $\delta$ be the maximum number of connected components when two vertices are removed from~$G$.
Then, $\egalPrice(G, 3) = \delta$.
\end{theorem}

We first show that~$\delta$ depends on the degrees of the two highest-degree vertices---we denote these degrees by~$\Delta_1(G)$ and~$\Delta_2(G)$, respectively---as well as whether these two vertices are adjacent to each other (it is possible that $\Delta_1(G) = \Delta_2(G)$).
In case of ties, we favor highest-degree vertices that are non-adjacent.

\begin{observation}
\label{lem:egal:delta}
Let $G$ be a tree.
Then, $\delta = \Delta_1(G) + \Delta_2(G) - 1$ if there exist two non-adjacent highest-degree vertices, and $\delta = \Delta_1(G) + \Delta_2(G) - 2$ otherwise.
\end{observation}

\begin{proof}
Observe that removing two vertices with degree $d_1$ and $d_2$ results in $d_1+d_2-1$ connected components if the two vertices are non-adjacent, and $d_1+d_2-2$ connected components otherwise.
Hence, to obtain the largest number of connected components, it is always optimal to choose two highest-degree vertices, favoring non-adjacent ones if possible.
\end{proof}

We next prove a lemma which shows that if a subtree is sufficiently valued by two agents, it can be divided between these agents so that each of them receives at least $1/\delta$ times her utility for the whole graph.

\begin{lemma}
\label{lem:egal:case1}
Let~$T^*$ be a subtree of~$G$ such that two agents~$i$ and~$j$ value~$T^*$ at least~$\frac{\Delta_1(G)}{\delta} \cdot u_i(G)$ and $\frac{\Delta_1(G)}{\delta} \cdot u_j(G)$, respectively.
Then, $T^*$ can be divided between the two agents so that agent~$i$ (resp., agent~$j$) receives a connected bundle worth at least~$u_i(G) / \delta$ (resp., $u_j(G) / \delta$).
\end{lemma}

\begin{proof}
Consider a modified instance where the agents' utilities are scaled so that each agent has value~$1$ for $T^*$.
Note that the maximum number of components when a vertex is removed from $T^*$ is $\Delta_1(T^*)$, and this number is~$2$ only when $T^*$ is a path.
Hence, by \Cref{thm:egal:2-agent-1connected}, $T^*$ can be divided between the two agents so that each agent receives a connected bundle worth at least $\frac{1}{\Delta_1(T^*)}$.
In the original instance, the same division gives agent~$i$ and $j$ a connected bundle worth at least $\frac{u_i(T^*)}{\Delta_1(T^*)}$ and $\frac{u_j(T^*)}{\Delta_1(T^*)}$, respectively.
Since $u_i(T^*) \geq \frac{\Delta_1(G)}{\delta} \cdot u_i(G)$ and $u_j(T^*) \geq \frac{\Delta_1(G)}{\delta} \cdot u_j(G)$ by assumption, and $\Delta_1(G) \geq \Delta_1(T^*)$ because $T^*$ is a subtree of~$G$, we can divide $T^*$ between agents~$i$ and~$j$ so that the agents receive utility at least $\frac{\Delta_1(G)}{\delta \cdot \Delta_1(T^*)} \cdot u_i(G) \geq \frac{u_i(G)}{\delta}$ and $\frac{\Delta_1(G)}{\delta \cdot \Delta_1(T^*)} \cdot u_j(G) \geq \frac{u_j(G)}{\delta}$, respectively.
\end{proof}

We now establish the main result for three agents.

\begin{proof}[Proof of \Cref{thm:egal:tree}]
A tree with~$3$ or $4$ vertices must be a path or a star, which will be covered in \Cref{thm:egal:n-agent-star} (for stars) and \Cref{thm:egal:n-agent-path} (for paths).
We therefore assume that the tree has at least~$5$ vertices and is not a star, so $\delta \geq 3$ and $\Delta_2(G) \geq 2$.
Let~$v_1$ and~$v_2$ denote the vertices with degree~$\Delta_1(G)$ and~$\Delta_2(G)$, respectively, whose removal results in~$\delta$ connected components.

\paragraph{Lower Bound:}
Let agents~$1$ and~$2$ have utility~$1$ for~$v_1$ and~$v_2$, respectively, and let agent~$3$ have utility~$1/\delta$ for an arbitrary vertex in each of the $\delta$ connected components resulting from the removal of~$v_1$ and~$v_2$.
The optimal egalitarian welfare overall is~$1$, while the optimal egalitarian welfare under connectivity is~$1/\delta$.
Hence, $\egalPrice(G, 3) \geq \delta$.

\paragraph{Upper Bound:}
By \Cref{prop:structured-instance}, it suffices to consider instances where each vertex is positively valued by at most one agent.
If some agent does not positively value any vertex, we may ignore that agent and apply \Cref{thm:egal:2-agent-1connected} to the remaining two agents.
Hence, assume that $u_i(G) > 0$ for every agent~$i$.
We prove the upper bound by devising \Cref{alg:egal:treecase1,alg:egal:treecase2}, which find a subtree~$T'$ of~$G$ such that $u_i(T') \geq u_i(G) / \delta$ for some agent~$i$ and the remaining subtree~$G \setminus T'$ can be satisfactorily divided between the other two agents.
We split the algorithm and proof into two cases, depending on whether $\delta = \Delta_1(G) + \Delta_2(G) - 1$ or $\delta = \Delta_1(G) + \Delta_2(G) - 2$.

\subparagraph{Case~1: $\delta = \Delta_1(G) + \Delta_2(G) - 1$.}
By Lemma~\ref{lem:egal:case1}, it suffices to show that there exists a subtree~$T'$ of~$G$ such that:
\begin{itemize}
\item $T'$ is worth at least $u_i(G) / \delta$ to some agent~$i$, and
\item $T'$ is worth at most $\frac{\Delta_2(G) - 1}{\delta} \cdot u_j(G)$ to each agent $j \neq i$ (that is, the complement $G\setminus T'$ is worth at least $\frac{\Delta_1(G)}{\delta} \cdot u_j(G)$).
\end{itemize}
After finding $T'$, we can allocate it to agent~$i$ and divide $G \setminus T'$ between the remaining agents to produce a connected allocation.

\begin{algorithm}[t]
\caption{Subroutine for a tree and $n = 3$ agents (Case~1)}
\label{alg:egal:treecase1}
\DontPrintSemicolon

\KwIn{An instance~$I = \langle [3], G, \utilityProfile \rangle$.
(Each vertex in~$G$ is positively valued by at most one agent.)}

Root the tree at~$v_1$, the vertex with degree $\Delta_1(G)$.\;
$\delta \gets \Delta_1(G)+\Delta_2(G)-1$\;
Take a top-level subtree~$T'$ of~$T$ such that $u_i(T') \geq \frac{u_i(G)}{\delta}$ for some agent~$i$.\; \label{alg:egal:tree:firstcase1}

\While{there exists a top-level subtree~$T''$ of~$T'$ worth at least $\frac{u_i(G)}{\delta}$ to some agent~$i$}{
	$T' \gets T''$\;
}

\Return{$T'$}
\end{algorithm}

We claim that \Cref{alg:egal:treecase1} finds such a subtree~$T'$.
Firstly, the subtree in \cref{alg:egal:tree:firstcase1} always exists because~$v_1$ can only be positively valued by at most one agent, and $\delta \geq \Delta_1(G)$.
When the while-loop terminates and the algorithm returns~$T'$, it must hold that $u_i(T') \geq \frac{u_i(G)}{\delta}$ for some~$i \in [3]$.
It remains to show that $u_j(T') \leq \frac{\Delta_2(G) - 1}{\delta} \cdot u_j(G)$ for all~$j \in [3] \setminus \{i\}$.
Since the while-condition is evaluated as false for~$T'$, we have for any top-level subtree~$T''$ of~$T'$ and any~$\ell \in [3]$ that $u_\ell(T'') < \frac{u_\ell(G)}{\delta}$.
Let $v$ denote the root vertex of~$T'$.
As~$T'$ has at most $\Delta_2(G) - 1$ top-level subtrees, it holds for all~$\ell \in [3]$ that $u_\ell(T' \setminus \{v\}) < \frac{\Delta_2(G) - 1}{\delta} \cdot u_\ell(G)$.

If $u_\ell(v) = 0$ for all~$\ell \in [3]$, then $u_\ell(T') = u_\ell(T' \setminus \{v\}) < \frac{\Delta_2(G) - 1}{\delta} \cdot u_\ell(G)$ for all~$\ell \in [3]$, so $T'$ satisfies the desired condition.
Otherwise, denote by~$i^*$ the agent with $u_{i^*}(v) > 0$; recall that the other two agents value vertex~$v$ at~$0$.
If $u_{i^*}(T') \geq \frac{u_{i^*}(G)}{\delta}$, then $u_j(T') = u_j(T' \setminus \{v\}) < \frac{\Delta_2(G) - 1}{\delta} \cdot u_i(G)$ for all~$j \in [3] \setminus \{i^*\}$, as desired.
Else, $u_{i^*}(T') < \frac{u_{i^*}(G)}{\delta} \leq \frac{\Delta_2(G) - 1}{\delta} \cdot u_{i^*}(G)$, where the latter inequality holds because $\Delta_2(G) \geq 2$.
Thus, some agent~$\ell \in [3] \setminus \{i^*\}$ values~$T'$ at least~$\frac{u_\ell(G)}{\delta}$, and for the agent (denoted as~$k$) other than~$i^*$ and~$\ell$, we have $u_k(T') = u_k(T' \setminus \{v\}) < \frac{\Delta_2(G) - 1}{\delta} \cdot u_k(G)$, as desired.

\subparagraph{Case~2: $\delta = \Delta_1(G) + \Delta_2(G) - 2$.}
\begin{algorithm}[t]
\caption{Subroutine for a tree and $n = 3$ agents (Case~2)}
\label{alg:egal:treecase2}
\DontPrintSemicolon

\KwIn{An instance~$I = \langle [3], G, \utilityProfile \rangle$.
(Each vertex in~$G$ is positively valued by at most one agent.)}

Root the tree at~$v_1$, the vertex with degree $\Delta_1(G)$.\;
$\delta \gets \Delta_1(G) + \Delta_2(G) - 2$\;
Take a top-level subtree~$T'$ of~$T$ such that $u_i(T') \geq \frac{u_i(G)}{\delta}$ for some agent~$i$.\; \label{alg:egal:tree:firstcase2}

\If{the root of~$T'$ has degree~$\Delta_2(G)$}{ \label{alg:egal:tree:specialcase}
	\eIf{there exists~$i\in [3]$ such that $u_i(T') \geq \frac{u_i(G)}{\delta}$ and $u_j(T') \leq \frac{\Delta_2(G) - 1}{\delta} \cdot u_j(G)$ for all~$j \in [3]\setminus\{i\}$}{
		\Return{$T'$} \label{alg:egal:tree:case2specialend}
	}{
		Set~$T''$ to be a top-level subtree of~$T'$ such that $u_i(T'') \geq \frac{u_i(G)}{\delta}$ for some agent~$i$.\; \label{alg:egal:tree:case2specialsubtree}
		$T' \gets T''$\;
	}
}

\While{there exists a top-level subtree~$T''$ of~$T'$ such that $u_i(T'') \geq \frac{u_i(G)}{\delta}$ for some agent~$i$}{
	$T' \gets T''$\;
}

\Return{$T'$}
\end{algorithm}

Similarly to Case~1, we find a subtree~$T'$ of~$G$ that is worth at least~$u_i(G) / \delta$ to some agent~$i$, such that $G \setminus T'$ is sufficiently valued by both remaining agents.
Root the tree~$G$ at~$v_1$, the vertex with degree~$\Delta_1(G)$.
Recall that in this case, the two highest-degree vertices are adjacent to each other.
This means that any vertex of distance at least two from~$v_1$ has degree at most~$\Delta_2(G) - 1$, since otherwise we would have two non-adjacent highest-degree vertices.

We claim that \Cref{alg:egal:treecase2} finds a desired subtree~$T'$.
To begin with, the (first) subtree~$T'$ in \cref{alg:egal:tree:firstcase2} always exists because~$v_1$ can only be positively valued by at most one agent, and $\delta \geq \Delta_1(G)$.
Let the root vertex of~$T'$ be denoted as $\widehat{v}$. 
We next distinguish cases based on whether $\widehat{v}$ has degree~$\Delta_2(G)$ or not.

We first address the special case in which \cref{alg:egal:tree:specialcase} is evaluated as true (i.e., $\deg(\widehat{v}) = \Delta_2(G)$) and the inner if-condition is met.
We allocate the subtree~$T'$ returned in \cref{alg:egal:tree:case2specialend} to the agent~$i$ identified in the inner if-condition; this agent receives utility at least~$u_i(G) / \delta$.
We now show that the remaining subtree~$T^* \coloneqq G \setminus T'$ can be divided between the remaining agents so that each of them receives a sufficiently valued connected bundle.
By the inner if-condition, for each agent $j \in [3]\setminus\{i\}$, we have $u_j(T') \leq \frac{\Delta_2(G) - 1}{\delta} \cdot u_j(G)$, which means that
\[
u_j(T^*) \geq \left( 1 - \frac{\Delta_2(G) - 1}{\delta} \right) \cdot u_j(G) = \frac{\Delta_1(G) - 1}{\delta} \cdot u_j(G).
\]
Since no two highest-degree vertices are non-adjacent, we have $\Delta_1(T^*) \leq \Delta_1(G) - 1$.
Combining this with (the proof of) \Cref{thm:egal:2-agent-1connected}, we deduce that $T^*$ can be divided between the agents in~$[3]\setminus \{i\}$ so that each agent $j\in [3]\setminus \{i\}$ receives a connected bundle worth at least $\frac{\Delta_1(G) - 1}{\delta \cdot \Delta_1(T^*)} \cdot u_j(G) \geq \frac{u_j(G)}{\delta}$.

Next, suppose that the inner if-condition is evaluated as false.
We claim that a top-level subtree~$T''$ of $T'$ such that $u_i(T'') \geq u_i(G) / \delta$ for some agent~$i$ always exists.
Suppose for contradiction that every subtree of~$T'$ is worth strictly less than $u_i(G) / \delta$ to each agent~$i$, and thus $u_i(T' \setminus \{\widehat{v}\}) < \frac{\Delta_2(G) - 1}{\delta} \cdot u_i(G)$.
If no agent values vertex $\widehat{v}$ positively, then $u_i(T') = u_i(T'\setminus\{\widehat{v}\}) < \frac{\Delta_2(G) - 1}{\delta} \cdot u_i(G)$ for all $i\in [3]$.
Since some agent $i\in [3]$ values $T'$ at least $u_i(G)/\delta$, the inner if-condition would have been evaluated as true for $T'$, a contradiction.
Hence, some agent $j$ positively values~$\widehat{v}$, and we may assume that $T'$ is worth more than $\frac{\Delta_2(G) - 1}{\delta} \cdot u_j(G) \ge u_j(G)$ to agent~$j$ (otherwise, we obtain the same contradiction as before).
Since each vertex can only be positively valued by one agent, and~$\widehat{v}$ is valued by agent~$j$, we know that~$T'$ is worth at most $\frac{\Delta_2(G) - 1}{\delta} \cdot u_k(G)$ and $\frac{\Delta_2(G) - 1}{\delta} \cdot u_\ell(G)$ to the other agents~$k$ and~$\ell$, respectively.
This means that~$T'$ satisfies the inner if-condition, again a contradiction.
Therefore, a top-level subtree~$T''$ of $T'$ such that $u_i(T'') \geq u_i(G) / \delta$ for some agent~$i$ exists.

We now address the general case in the while-loop.
Note that the root of each subtree considered from now on (as $T''$) has degree at most~$\Delta_2(G) - 2$; otherwise there would be two non-adjacent highest-degree vertices.
We show that the algorithm terminates when the following conditions on~$T'$ are met:
\begin{itemize}
\item $T'$ is worth at least $u_i(G) / \delta$ to some agent~$i$, and
\item $T'$ is worth at most $\frac{\Delta_2(G) - 2}{\delta} \cdot u_j(G)$ to each agent~$j \neq i$ (that is, the complement $G\setminus T'$ is worth at least $\frac{\Delta_1(G)}{\delta} \cdot u_j(G)$).
\end{itemize}
Under these conditions, we can allocate~$T'$ to agent~$i$, and by \Cref{lem:egal:case1}, divide $G\setminus T'$ between the remaining agents so that each agent $j \in [3] \setminus \{i\}$ receives a connected bundle worth at least $u_j(G)/\delta$.

When the algorithm returns~$T'$, it holds that $u_\ell(T') \geq \frac{u_\ell(G)}{\delta}$ for some~$\ell \in [3]$.
Since the while-condition is evaluated as false for~$T'$, for any top-level subtree~$T''$ of~$T'$ and any~$i \in [3]$, we have $u_i(T'') < \frac{u_i(G)}{\delta}$.
Denote the root of~$T'$ by~$v$.
As~$T'$ has at most $\Delta_2(G) - 2$ top-level subtrees, for each~$i \in [3]$, it holds that $u_i(T' \setminus \{v\}) \leq \frac{\Delta_2(G) - 2}{\delta} \cdot u_i(G)$.
If $u_i(v) = 0$ for all~$i \in [3]$, meaning that $u_i(T') = u_i(T' \setminus \{v\}) \leq \frac{\Delta_2(G) - 2}{\delta} \cdot u_i(G)$ for all~$i \in [3]$, then $T'$ satisfies the desired conditions.
Otherwise, denote by~$i^*$ the agent with $u_{i^*}(v) > 0$; note that the other two agents value vertex~$v$ at~$0$.
If $u_{i^*}(T') \geq \frac{u_{i^*}(G)}{\delta}$, then $u_j(T') = u_j(T' \setminus \{v\}) \leq \frac{\Delta_2(G) - 2}{\delta} \cdot u_j(G)$ for all~$j \in [3] \setminus \{i^*\}$, and $T'$ satisfies the desired conditions.
Else, $u_{i^*}(T') < \frac{u_{i^*}(G)}{\delta}$, so $i^* \neq \ell$.
Since agent~$\ell$ values~$T'$ at least~$\frac{u_\ell(G)}{\delta} > 0$ and $u_\ell(v) = 0$, it must hold that $\Delta_2(G) - 2 > 0$, that is, $\Delta_2(G) \ge 3$.
This implies that $u_{i^*}(T') < \frac{u_{i^*}(G)}{\delta} \leq \frac{\Delta_2(G) - 2}{\delta} \cdot u_{i^*}(G)$.
Also, for the agent (denoted as~$k$) other than~$i^*$ and~$\ell$, it holds that $u_k(T') = u_k(T' \setminus \{v\}) \leq \frac{\Delta_2(G) - 2}{\delta} \cdot u_k(G)$, so $T'$ again satisfies the desired conditions.
\end{proof}

\subsection{Any Number of Agents}

Next, we address the general case where the number of agents can be arbitrary.
We begin with stars and paths.

\begin{theorem}
\label{thm:egal:n-agent-star}
Let $G$ be a star and $n\ge 2$.
Then, $\egalPrice(G, n) = m - n + 1$.
\end{theorem}

\begin{proof}
By \Cref{lem:egal:n-agent-any-upper}, the egalitarian PoC is at most $m-n+1$.
To establish tightness, we consider an instance where one agent values $m-n+1$ of the leaf vertices at $\frac{1}{m-n+1}$ each, and each of the $n-1$ remaining agents values a distinct vertex at~$1$.
The optimal egalitarian welfare overall is~$1$, while the optimal egalitarian welfare under connectivity is $\frac{1}{m-n+1}$.
\end{proof}

\begin{theorem}
\label{thm:egal:n-agent-path}
Let $G$ be a path and $n\ge 2$.
Then,
\[
\egalPrice(G, n) = \begin{cases}
m - n + 1 & \text{ if } n \leq m < 2n - 1; \\
n & \text{ if } 2n - 1 \leq m.
\end{cases}
\]
\end{theorem}

\begin{proof}
We establish the lower bound and the upper bound in turn.

\paragraph{Lower Bound:}
Consider a path with vertices~$1, 2, \dots, m$, ordered from left to right.
For the case where $n \leq m < 2n - 1$, we let the agents have the following utilities:
\begin{itemize}
\item Agent~$1$ values~$m -n + 1$ vertices equally at~$\frac{1}{m - n + 1}$ each, and moreover, none of these vertices are adjacent to each other.
Note that this is possible since $m \ge 2(m-n+1)-1$, which is implied by $m < 2n-1$.

\item The remaining~$n-1$ agents value the~$n-1$ vertices that are not valued by agent~$1$.
In particular, each of these agents values a distinct vertex among these vertices at~$1$.
\end{itemize}
The optimal egalitarian welfare overall is~$1$, while the optimal egalitarian welfare under connectivity is $\frac{1}{m-n+1}$, so  $\egalPrice(G, n) \geq m - n + 1$.

For the case where $m \geq 2n-1$, let the agents have the following utilities:
\begin{itemize}
\item For each agent~$i \in [n-1]$, let $u_i(2i) = 1$ and $u_i(j) = 0$ for all $j \in [m] \setminus \{2i\}$.

\item For agent~$n$, let $u_n(2i - 1) = \frac{1}{n}$ for each~$i \in [n]$, and let the agent value all other vertices at~$0$.
\end{itemize}
The optimal egalitarian welfare overall is~$1$, while the optimal egalitarian welfare under connectivity is $\frac{1}{n}$, so $\egalPrice(G, n) \geq n$.

\begin{algorithm}[t]
\caption{For a path and any number of agents $n$}
\label{alg:egal:moving-knife}
\DontPrintSemicolon

\KwIn{An instance~$I = \langle N, G, \utilityProfile \rangle$.
(Each vertex in~$G$ is valued by at most one agent.)}
\KwOut{A connected allocation~$\alloc$ with $\egalSW(\alloc) \geq \frac{1}{n} \cdot \egalOPT(I)$.}

$\alloc = (M_1, \dots, M_n) \gets (\emptyset, \dots, \emptyset)$\;

\While{$|N| \geq 2$}{
	$B \gets \emptyset$\;
	Process the vertices along the path from left to right and add them one at a time to bundle~$B$ until $u_i(B) \geq \frac{1}{n} \cdot \egalOPT(I)$ for some $i\in N$.\; \label{alg:egal:moving-knife:add-vertex}
	$M_i \gets B$\;
	$N \gets N \setminus \{i\}$\;
        $G \gets G \setminus B$\;
}

Give all remaining vertices to the last agent in~$N$, and update~$\alloc$ accordingly.

\Return{$\alloc$}
\end{algorithm}

\paragraph{Upper Bound:}
For $n\leq m< 2n-1$, the upper bound follows from~\Cref{lem:egal:n-agent-any-upper}, so we focus on the case where $m\geq 2n-1$.
By \Cref{prop:structured-instance}, it suffices to show that given any instance~$I = \langle N, G, \utilityProfile \rangle$ in which each vertex of the path~$G$ is positively valued by at most one agent, there exists a connected allocation such that each agent $i\in N$ receives utility at least $\frac{1}{n} \cdot \egalOPT(I)$.
Our algorithm is a discretized version of the well-known \emph{moving-knife procedure} \citep{DubinsSp61}; the pseudocode can be found as \Cref{alg:egal:moving-knife}.
Since each vertex is positively valued by at most one agent, whenever a vertex is added to bundle~$B$ in \cref{alg:egal:moving-knife:add-vertex}, at most one agent's utility for~$B$ strictly increases.
As a result, at any point during the execution of the algorithm, at most one agent finds $B$ to be worth at least $\frac{1}{n} \cdot \egalOPT(I)$.
Each agent removed from the instance in the while-loop receives a contiguous bundle of vertices and gets utility at least $\frac{1}{n} \cdot \egalOPT(I)$.
Finally, the last agent values all remaining vertices at least $u_i(G) - \frac{n-1}{n}\cdot\egalOPT(I) \ge \egalOPT(I) - \frac{n-1}{n}\cdot\egalOPT(I) = \frac{1}{n}\cdot\egalOPT(I)$.
\end{proof}

We end this section by establishing the egalitarian PoC for cycles.

\begin{theorem}
\label{thm:egal:n-agent-cycle}
Let $G$ be a cycle and $n\ge 2$.
Then,
\[
\egalPrice(G, n) =
\begin{cases}
m - n + 1 & \text{ if } n \leq m < 2n-2; \\
n - 1 & \text{ if } 2n-2 \leq m < n^2; \\
n & \text{ if } n^2 \leq m.
\end{cases}
\]
\end{theorem}

\begin{proof}
We distinguish three cases as follows.

\paragraph{Case~1: $n \leq m < 2n - 2$.}
The upper bound of $m-n+1$ follows from \Cref{lem:egal:n-agent-any-upper}.
For the lower bound, consider a cycle such that agent~$1$ values $m-n+1$ vertices equally at $\frac{1}{m-n+1}$ each, and moreover, none of these vertices are adjacent to each other---note that this is possible since $m\ge 2(m-n+1)$, which is implied by $m < 2n-2$.
The optimal egalitarian welfare overall is~$1$, while the optimal egalitarian welfare under connectivity is $\frac{1}{m-n+1}$, so $\egalPrice(G, n) \geq m - n + 1$.

\paragraph{Case~2: $2n - 2 \leq m < n^2$.}
\begin{figure}[t]
\centering
\tikzstyle{vertex}=[circle,draw]
\begin{tikzpicture}[on grid]
\node[vertex,label=above:{$a_1$}] (a1-1) {};
\node[vertex,label=above:{$a_2$}] (a2) [right=of a1-1] {};
\node[vertex,label=above:{$a_1$}] (a1-2) [right=of a2] {};
\node[vertex,label=above:{$a_3$}] (a3) [right=of a1-2] {};
\node (dots) [right=of a3] {$\cdots$};
\node[vertex,label=above:{$a_1$}] (a1-n-2) [right=of dots] {};
\node[vertex,label=above:{$a_{n-1}$}] (an-1) [right=of a1-n-2] {};
\node[vertex,label=above:{$a_1$}] (a1-n-1) [right=of an-1] {};
\node[vertex,label=above:{$a_n$}] (an) [right=of a1-n-1] {};
\node (dots2) [right=of an] {$\cdots$};
\node[vertex] (vm) [right=of dots2] {};

\draw (a1-1) -- (a2) -- (a1-2) -- (a3) -- (dots) -- (a1-n-2) -- (an-1) -- (a1-n-1) -- (an) -- (dots2) -- (vm);
\path (vm) edge [bend left=20] (a1-1);
\end{tikzpicture}
\caption{A cycle with $2n - 2 \leq m < n^2$ vertices.
Agent~$1$ values~$n-1$ non-adjacent vertices (labeled $a_1$) at~$\frac{1}{n-1}$ each, while each remaining agent~$i\in\{2,\dots,n\}$ values her designated vertex $a_i$ at~$1$.}
\label{fig:cycle}
\end{figure}

We first prove the lower bound.
Consider a cycle where the first $2n-2$ vertices alternate between an item valued by agent~$1$ and an item valued by each agent in the set $\{2,\dots,n\}$ (see \Cref{fig:cycle}).
Agent~$1$ values each of her items at $\frac{1}{n-1}$, and each of agents~$2,\dots,n$ values her item at $1$.
The optimal egalitarian welfare overall is~$1$, while the optimal egalitarian welfare under connectivity is $\frac{1}{n-1}$, so $\egalPrice(G, n) \geq n-1$.

We now turn to the upper bound.
By \Cref{prop:structured-instance}, it suffices to show that given any instance~$I = \langle N, G, \utilityProfile \rangle$ in which each vertex of the cycle~$G$ is positively valued by at most one agent, there exists a connected allocation such that each agent receives utility at least~$\frac{\egalOPT(I)}{n-1}$.
Since the total number of items is less than $n^2$, there exists an agent who positively values at most~$n-1$ items.
By giving one such agent her most preferred item, she receives utility at least~$\frac{\egalOPT(I)}{n-1}$.
Ignoring that item and that agent, the remaining graph is a path with at least~$2n-3$ vertices, to be divided among~$n-1$ agents.
Since each of these $n-1$ agents does not value the allocated item, the optimal egalitarian welfare in the reduced instance (with $n-1$ agents) is at least that in the original instance (with $n$ agents). 
From the proof of \Cref{thm:egal:n-agent-path}, we know that when there are~$n-1$ agents and $m \geq 2(n-1) - 1 = 2n-3$ vertices on a path such that each vertex is positively valued by at most one agent, there exists a connected allocation where each agent receives at least $\frac{1}{n-1}$ of the optimal egalitarian welfare overall.
This completes the proof for this case.

\paragraph{Case~3: $m \geq n^2$.}
The upper bound follows directly from \Cref{thm:egal:n-agent-path}, as we can remove an arbitrary edge of the cycle and apply the corresponding algorithm in the proof of \Cref{thm:egal:n-agent-path} to the resulting path.

For the lower bound, consider first the case $m = n^2$, and assume that the vertices of the cycle are $1,2,\dots,n^2$ in this order.
For each $i,j\in [n]$, let agent~$i$ value vertex $(j-1)n + i$ at $1/n$.
The optimal egalitarian welfare overall is~$1$.
On the other hand, for an agent to receive at least two items that she values positively in a connected allocation, she must receive at least $n+1$ items in total.
Hence, for every agent to receive at least two items that she values positively, the total number of items must be at least $(n+1)n \ge n^2$, which is impossible.
This means that in any connected allocation, some agent receives at most one item that she values positively, and the egalitarian welfare under connectivity is at most $1/n$.
It follows that the egalitarian PoC is at least $n$.

To extend the example to $m > n^2$, simply add $m-n^2$ dummy items that no agent values positively, and apply the same argument ignoring the dummy items.
\end{proof}

\section{Utilitarian Price of Connectivity}
\label{sec:util}

In this section, we turn our attention to the utilitarian price of connectivity, which measures the worst-case loss of utilitarian welfare due to connectivity constraints.
Note that an optimal (unconstrained) utilitarian allocation simply assigns each item to an agent who values it most.

\subsection{Two Agents}

We begin by determining the utilitarian PoC for instances with two agents.
An important quantity to consider for this purpose is the difference between the two agents' utilities for each item, that is, $\Delta_v \coloneqq u_2(v) - u_1(v)$ for each~$v \in V$.

Similarly to the egalitarian PoC, the utilitarian PoC is strictly higher than $1$ when a single edge is removed from a complete graph, but does not increase further when additional disjoint edges are removed (with the exception of the $5$-item case).

\begin{theorem}
\label{thm:util:complete}
Let $G$ be a complete graph with a non-empty matching removed.
Then, $\utilPrice(G, 2) = \frac{4}{3}$ if~$G$ is
\begin{itemize}
\item $L_3$, which is $K_3$ with an edge removed, or
\item $L_5$, which is $K_5$ with two disjoint edges removed,
\end{itemize}
and $\utilPrice(G, 2) = \frac{2m-4}{2m-5}$ otherwise.
\end{theorem}

\begin{proof}
When $m = 3$, the utilitarian PoC of $4/3$ for $L_3$ follows from \Cref{thm:util:2-agent:tree} for trees.
When $m = 4$, the utilitarian PoC of $4/3$ for $K_4$ with two disjoint edges removed follows from \Cref{thm:util:2-agent:cycle} for cycles.
This implies that for $K_4$ with a single edge removed, the utilitarian PoC is at most $4/3$.
We now show that the utilitarian PoC is also at least $4/3$, and therefore exactly $4/3$.
Consider the instance~$I$ where agent~$1$ values the two vertices without an edge between them at~$1/2$ each, and agent~$2$ values the two remaining vertices at~$1/2$ each.
We have $\utilOPT(I) = 2$ and $\max_{\alloc \in \connectedAllocsOf{I}} \utilSW(\alloc) = 3/2$, so the utilitarian PoC is at least~$4/3$.

For the remainder of the proof, let $m \geq 5$.
We establish the lower bound and the upper bound in turn.

\paragraph{Lower Bound (for $m \geq 5$):}
First, consider $L_5$, and denote the missing edges by $\{v_{1,1}, v_{1,2}\}$ and $\{v_{2,1}, v_{2,2}\}$.
For~$i \in [2]$, let agent~$i$ value vertices $v_{i,1}$ and $v_{i,2}$ at~$1/2$ each, and let the final vertex~$v_5$ be non-valued.
Observe that either agent can obtain both of her valued vertices by taking~$v_5$, but if one agent does so, then the other agent can only obtain one of her valued vertices.
Hence, the optimal utilitarian welfare under connectivity is at most~$3/2$, while the optimal utilitarian welfare overall is~$2$, and so the utilitarian PoC is at least~$4/3$.

Consider now any graph that is not~$L_5$.
Let~$v_1$ and~$v_2$ be two vertices without an edge between them, let agent~$1$ value each of these vertices at $1/2$, and let agent~$2$ value each of the remaining $m-2$ vertices at $\frac{1}{m-2}$.
A utilitarian-optimal connected allocation assigns to agent~$1$ both of her valued vertices along with one of agent~$2$'s valued vertices, and assigns to agent~$2$ the remaining vertices.
This leads to a utilitarian welfare of $1 + \frac{m-3}{m-2} = \frac{2m-5}{m-2}$.
On the other hand, the optimal utilitarian welfare overall is~$2$, so the utilitarian PoC is at least $\frac{2m-4}{2m-5}$.

\paragraph{Upper Bound (for $m \geq 5$):}
To prove the upper bound, we show that among all possible instances, those described in the lower bound examples maximize the utilitarian welfare ratio $\frac{\utilOPT(I)}{\max_{\alloc \in \connectedAllocsOf{I}} \utilSW(\alloc)}$.
Observe that every bundle of size either $1$ or at least~$3$ is connected.
Hence, the only scenario in which there does not exist a connected allocation where each vertex is allocated to an agent who values it most is when there exist vertices $v_{1,1},v_{1,2}$ that are not connected by an edge, such that one agent (say, agent~$1$) values $v_{1,1},v_{1,2}$ more than agent~$2$, and agent~$2$ values every remaining vertex more than agent~$1$.
Let us therefore proceed with these assumptions, and consider an instance~$I$.

We first handle the special case where the graph is~$L_5$. 
Let $v_{2,1},v_{2,2}$ be the other pair of vertices not connected by an edge, and let $v_5$ be the remaining vertex.
We have $\utilOPT(I)=u_1(\{v_{1,1},v_{1,2}\}) + u_2(\{v_{2,1},v_{2,2},v_5\})$. 
Now, consider the connected allocation where agent $2$ is allocated $v_{2,1}$, $v_{2,2}$, $v_5$, along with a vertex $v\in \{v_{1,1},v_{1,2}\}$ that minimizes $u_1(v)-u_2(v)$; assume without loss of generality that $v = v_{1,1}$. 
In this allocation, agent $1$ receives only $v_{1,2}$, which means that $\max_{\alloc \in \connectedAllocsOf{I}} \utilSW(\alloc)\geq u_2(\{v_{2,1},v_{2,2},v_5, v_{1,1}\}) + u_1(v_{1,2})$.
Hence,
\begin{align*}
\frac{\utilOPT(I)}{\max_{\alloc \in \connectedAllocsOf{I}} \utilSW(\alloc)}&\leq \frac{u_2(\{v_{2,1},v_{2,2},v_5\}) + u_1(\{v_{1,1},v_{1,2}\})}{u_2(\{v_{2,1},v_{2,2},v_5, v_{1,1}\}) + u_1(v_{1,2})}\\
&=\frac{1-u_2(\{v_{1,1},v_{1,2}\})+u_1(\{v_{1,1},v_{1,2}\})}{1-u_2(v_{1,2})+u_1(v_{1,2})} \\
&=\frac{1 + (u_1(v_{1,1}) - u_2(v_{1,1})) + (u_1(v_{1,2}) - u_2(v_{1,2}))}{1 + (u_1(v_{1,2}) - u_2(v_{1,2}))}\\
&\le\frac{1 + (u_1(v_{1,1}) - u_2(v_{1,1})) + (u_1(v_{1,2}) - u_2(v_{1,2}))}{1 + \frac{1}{2}(u_1(v_{1,1}) - u_2(v_{1,1})) + \frac{1}{2}(u_1(v_{1,2}) - u_2(v_{1,2}))}\\
&= \frac{1 + 2x}{1 +x},
\end{align*}
where $x \coloneqq \frac{1}{2}(u_1(\{v_{1,1},v_{1,2}\}) - u_2(\{v_{1,1},v_{1,2}\}))$.
Since $x\in [-1/2, 1/2]$, it follows that $\frac{1+2x}{1+x}\le \frac{4}{3}$, and so the utilitarian PoC of the graph~$L_5$ is at most $4/3$.

It remains to consider graphs that are not~$L_5$.
For simplicity of notation, we rename $v_{1,1},v_{1,2}$ to $v_1,v_2$, respectively, so agent~$1$ values $v_1,v_2$ more than agent~$2$ while agent~$2$ values every remaining vertex more than agent~$1$.
Let $v' \in \argmax_{v \in V \setminus \{v_1,v_2\}} (u_1(v)-u_2(v))$, so 
\begin{align*}
u_1(v') - u_2(v') &\ge \frac{1}{m-2}\left(u_1(V\setminus\{v_1,v_2\}) - u_2(V\setminus\{v_1,v_2\})\right) \\
&= \frac{1}{m-2}((1-u_1(\{v_1,v_2\})) - (1-u_2(\{v_1,v_2\}))) \\
&= \frac{1}{m-2}(u_2(\{v_1,v_2\}) - u_1(\{v_1,v_2\})).
\end{align*}
Consider the allocation that gives agent~$1$ vertices~$v_1$, $v_2$, and $v'$, and agent~$2$ all of the remaining vertices.
Since the graph is not $L_5$, this allocation is connected.
Hence, we have
\begin{align*}
\frac{\utilOPT(I)}{\max_{\alloc \in \connectedAllocsOf{I}} \utilSW(\alloc)} &\leq \frac{u_1(\{v_1, v_2\}) + u_2(V \setminus \{v_1,v_2\})}{u_1(\{v_1,v_2,v'\}) + u_2(V \setminus \{v_1,v_2,v'\})} \\
&= \frac{1 + (u_1(\{v_1, v_2\}) - u_2(\{v_1, v_2\}))}{1 + (u_1(\{v_1, v_2\}) - u_2(\{v_1, v_2\})) + (u_1(v') - u_2(v'))} \\
&\leq \frac{1 + (u_1(\{v_1, v_2\}) - u_2(\{v_1, v_2\}))}{1 + (u_1(\{v_1, v_2\}) - u_2(\{v_1, v_2\})) + \frac{1}{m-2}(u_2(\{v_1, v_2\}) - u_1(\{v_1, v_2\}))} \\
&\leq \frac{1 + (u_1(\{v_1, v_2\}) - u_2(\{v_1, v_2\}))}{1 + \frac{m-3}{m-2}(u_1(\{v_1, v_2\}) - u_2(\{v_1, v_2\}))} \\
&= \frac{1+y}{1+\frac{m-3}{m-2}\cdot y},
\end{align*}
where $y \coloneqq u_1(\{v_1, v_2\}) - u_2(\{v_1, v_2\})$.
Since $y\in [0,1]$, it follows that $\frac{1+y}{1+\frac{m-3}{m-2}\cdot y} \le \frac{2m-4}{2m-5}$, so the utilitarian PoC is at most $\frac{2m-4}{2m-5}$.
\end{proof}

We next consider complete bipartite graphs.
Interestingly, if one side of the graph has two vertices, then the utilitarian PoC is constant, regardless of how many vertices there are on the other side.
Since the case where one side of the graph has a single vertex corresponds to a star, which will be covered by~\Cref{thm:util:2-agent:tree}, we assume that both sides have at least two vertices each. 

\begin{theorem}
\label{thm:util:2-agent:complete-bipartite}
Let $G$ be a complete bipartite graph with vertex sets $V_1$ and $V_2$, where $|V_1|,|V_2|\ge 2$.
Then,
\[
\utilPrice(G, 2) = \begin{cases}
\frac{4}{3} & \text{ if } |V_1| = 2 \text{ or } |V_2| = 2; \\
\frac{2}{2 - \frac{1}{|V_1|} - \frac{1}{|V_2|}} & \text{ otherwise}.
\end{cases}
\]
\end{theorem}

\begin{proof}
We establish the lower bound and the upper bound in turn.

\paragraph{Lower Bound:}
For the special case where $|V_1| = 2$ or $|V_2| = 2$, assume without loss of generality that $|V_1| = 2$.
Let agent~$1$ have utility $1/2$ for each item in~$V_1$, and agent~$2$ have utility $1/|V_2|$ for each item in~$V_2$.
The optimal utilitarian welfare overall is $2$, while the optimal utilitarian welfare under connectivity is $3/2$, achieved by giving agent~$1$ one item from $V_1$ and none from $V_2$.
Hence, the utilitarian PoC is at least $4/3$.

For the general case where $|V_1|, |V_2| \geq 3$, let agent~$1$ have utility $1/|V_1|$ for each item in $V_1$, and agent~$2$ have utility $1/|V_2|$ for each item in $V_2$.
The optimal utilitarian welfare overall is $2$, while the optimal utilitarian welfare under connectivity is $2 - 1/|V_1| - 1/|V_2|$, achieved by giving each agent all but one of her valued items along with one of the other agent's valued items.
Hence, the utilitarian PoC is at least $\frac{2}{2-1/|V_1|-1/|V_2|}$.

\paragraph{Upper Bound:}
Consider an instance~$I$.
If for each~$i \in [2]$, there exists a pair of distinct vertices~$v_i, v'_i \in V_i$ such that $u_1(v_i) \geq u_2(v_i)$ and $u_1(v'_i) \leq u_2(v'_i)$, then we can obtain a connected allocation that achieves the optimal (unconstrained) utilitarian welfare by giving vertices~$v_1, v_2$ (resp., $v'_1, v'_2$) to agent~$1$ (resp., agent~$2$) and each remaining vertex to an agent who values it most.
Therefore, we may assume that $u_1(v) > u_2(v)$ for every $v\in V_1$.

If $|V_1| = |V_2| = 2$, the graph is a cycle, so the utilitarian PoC of $4/3$ follows from \Cref{thm:util:2-agent:cycle}.
We now perform a case distinction based on $|V_2|$.

\subparagraph{Case~1: $|V_2| = 2$ (and $|V_1| \ge 3$).}
If there exists $v \in V_2$ such that $u_1(v) \geq u_2(v)$, then there is a connected allocation that assigns each vertex to an agent who values it most.
Therefore, assume that $u_1(v) < u_2(v)$ for every $v\in V_2$.
Let $v' \in \argmax_{v \in V_2} (u_1(v) - u_2(v))$.
Since the allocation that gives $V_1\cup\{v'\}$ to agent~$1$ and $V_2\setminus\{v'\}$ to agent~$2$ is connected, we have
\begin{align*}
\frac{\utilOPT(I)}{\max_{\alloc \in \connectedAllocsOf{I}} \utilSW(\alloc)} &\leq \frac{u_1(V_1) + u_2(V_2)}{u_1(V_1) + u_1(v') + u_2(V_2 \setminus \{v'\})} \\
&= \frac{u_1(V_1) + u_2(V_2)}{u_1(V_1) + u_2(V_2) + (u_1(v') - u_2(v'))} \\
&\leq \frac{u_1(V_1) + u_2(V_2)}{u_1(V_1) + u_2(V_2) + \frac{1}{2}(u_1(V_2) - u_2(V_2))} \\
&= \frac{u_1(V_1) + u_2(V_2)}{u_1(V_1) + \frac{1}{2}u_1(V_2) + \frac{1}{2}u_2(V_2)} \\
&= \frac{u_1(V_1) + u_2(V_2)}{\frac{1}{2} + \frac{1}{2}u_1(V_1) + \frac{1}{2}u_2(V_2)} \\
&= \frac{2u_1(V_1) + 2u_2(V_2)}{1 + u_1(V_1) + u_2(V_2)} \\
&= 2 - \frac{2}{1 + u_1(V_1) + u_2(V_2)} \\
&\le \frac{4}{3},
\end{align*}
where the last inequality holds because $u_1(V_1),u_2(V_2) \le 1$.
Hence, the utilitarian PoC is at most $4/3$ in this case.

\subparagraph{Case~2: $|V_2| \geq 3$ (and $|V_1| \geq 2$).}
Recall that $u_1(v) > u_2(v)$ for every~$v \in V_1$.
We first deal with the subcase where there also exists a non-empty subset $V' \subset V_2$ such that $u_1(v) \geq u_2(v)$ for each~$v \in V'$.
Let $v' \in \argmax_{v \in V_1} (u_2(v) - u_1(v))$.
Consider the connected allocation where agent~$2$ receives each vertex~$v \in V_2$ such that $u_2(v) > u_1(v)$ along with vertex~$v'$, and agent~$1$ receives every other vertex.
We have
\begin{align*}
\frac{\utilOPT(I)}{\max_{\alloc \in \connectedAllocsOf{I}} \utilSW(\alloc)}
&\leq \frac{u_1(V_1) + u_1(V') + u_2(V_2 \setminus V')}{u_1(V_1) - u_1(v') + u_1(V') + u_2(v') + u_2(V_2 \setminus V')} \\
&= \frac{u_1(V_1) + u_1(V') + u_2(V_2 \setminus V')}{u_1(V_1) + u_1(V') + u_2(V_2 \setminus V') + (u_2(v') - u_1(v'))} \\
&\le \frac{u_1(V_1) + u_1(V') + u_2(V_2 \setminus V')}{u_1(V_1) + u_1(V') + u_2(V_2 \setminus V') + \frac{1}{|V_1|}(u_2(V_1) - u_1(V_1))} \\
&= \frac{u_1(V_1) + u_1(V') + u_2(V_2 \setminus V')}{\frac{|V_1|-1}{|V_1|}\cdot u_1(V_1) + u_1(V') + u_2(V_2 \setminus V') + \frac{1}{|V_1|}\cdot u_2(V_1) } \\
&= \frac{u_1(V_1) + u_2(V_2) + (u_1(V') - u_2(V'))}{\frac{|V_1|-1}{|V_1|}\cdot u_1(V_1) + u_2(V_2) + \frac{1}{|V_1|}\cdot u_2(V_1) + (u_1(V') - u_2(V')) } \\
&\le \frac{u_1(V_1) + u_2(V_2)}{\frac{|V_1|-1}{|V_1|}\cdot u_1(V_1) + u_2(V_2) + \frac{1}{|V_1|}\cdot u_2(V_1) } \\
&= \frac{u_1(V_1) + u_2(V_2)}{\frac{|V_1|-1}{|V_1|}\cdot u_1(V_1) + \frac{|V_1|-1}{|V_1|}\cdot u_2(V_2) + \frac{1}{|V_1|}\cdot u_2(V_2) + \frac{1}{|V_1|}\cdot u_2(V_1) } \\
&= \frac{u_1(V_1) + u_2(V_2)}{\frac{|V_1|-1}{|V_1|}\cdot u_1(V_1) + \frac{|V_1|-1}{|V_1|}\cdot u_2(V_2) + \frac{1}{|V_1|} } \\
&= \frac{\frac{|V_1|}{|V_1|-1}(u_1(V_1) + u_2(V_2))}{ u_1(V_1) +  u_2(V_2) + \frac{1}{|V_1|-1} } \\
&= \frac{\frac{|V_1|}{|V_1|-1}}{ 1 + \frac{1}{|V_1|-1}\cdot\frac{1}{u_1(V_1) + u_2(V_2)} } \\
&\le \frac{\frac{|V_1|}{|V_1|-1}}{ 1 + \frac{1}{|V_1|-1}\cdot\frac{1}{2} } \\
&= \frac{2}{2 - \frac{1}{|V_1|}}.
\end{align*}
If $|V_1|=2$, the expression $\frac{2}{2-1/|V_1|}$ becomes $4/3$.
Else, the expression is strictly less than $\frac{2}{2 - 1/|V_1| - 1/|V_2|}$.

Next, we address the remaining subcase  where $u_2(v) > u_1(v)$ for all $v\in V_2$.
If $|V_1| = 2$, this is already covered by the analysis in Case~1, so we assume that $|V_1| \geq 3$.
Let $v'_1 \in \argmax_{v \in V_1} (u_2(v) - u_1(v))$ and $v'_2 \in \argmax_{v \in V_2} (u_1(v) - u_2(v))$.
Consider the connected allocation where agent~$1$ receives~$v'_2$ and $V_1 \setminus \{v'_1\}$, while agent~$2$ receives~$v'_1$ and $V_2 \setminus \{v'_2\}$.
We have
\begin{align*}
\frac{\utilOPT(I)}{\max_{\alloc \in \connectedAllocsOf{I}} \utilSW(\alloc)}
&\leq\frac{u_1(V_1) + u_2(V_2)}{u_1(v'_2) + u_1(V_1 \setminus \{v'_1\}) + u_2(v'_1) + u_2(V_2 \setminus \{v'_2\})} \\
&= \frac{u_1(V_1) + u_2(V_2)}{u_1(V_1) + u_2(V_2) + (u_1(v'_2) - u_2(v'_2)) + (u_2(v'_1) - u_1(v'_1))} \\
&\le \frac{u_1(V_1) + u_2(V_2)}{u_1(V_1) + u_2(V_2) + \frac{1}{|V_2|}(u_1(V_2) - u_2(V_2)) + \frac{1}{|V_1|}(u_2(V_1) - u_1(V_1))} \\
&= \frac{u_1(V_1) + u_2(V_2)}{(1-\frac{1}{|V_1|})\cdot u_1(V_1) + \frac{1}{|V_2|}\cdot u_1(V_2) + (1-\frac{1}{|V_2|})\cdot u_2(V_2) + \frac{1}{|V_1|}\cdot u_2(V_1) } \\
&= \frac{u_1(V_1) + u_2(V_2)}{(1 - \frac{1}{|V_1|} - \frac{1}{|V_2|})(u_1(V_1) + u_2(V_2)) + \frac{1}{|V_1|} + \frac{1}{|V_2|}} \\
&= \frac{1}{(1 - \frac{1}{|V_1|} - \frac{1}{|V_2|}) + (\frac{1}{|V_1|} + \frac{1}{|V_2|})\frac{1}{u_1(V_1) + u_2(V_2)}} \\
&\le \frac{1}{(1 - \frac{1}{|V_1|} - \frac{1}{|V_2|}) + (\frac{1}{|V_1|} + \frac{1}{|V_2|})\frac{1}{2}} \\
&= \frac{2}{2 - \frac{1}{|V_1|} - \frac{1}{|V_2|}}.
\end{align*}
This shows that the utilitarian PoC is at most $\frac{2}{2 - 1/|V_1| - 1/|V_2|}$, and completes the proof.
\end{proof}

For trees, we find that the utilitarian PoC depends only on the number of vertices.
This is in contrast to the egalitarian PoC, which depends on the maximum degree among the vertices in the tree (\Cref{thm:egal:2-agent-1connected}).

\begin{theorem}
\label{thm:util:2-agent:tree}
Let $G$ be a tree.
Then, $\utilPrice(G, 2) = \frac{2 (m-1)}{m}$.
\end{theorem}

\begin{proof}
We establish the lower bound and the upper bound in turn.

\paragraph{Lower Bound:}
Color the vertices of the tree in two colors (say, red and blue) so that no two adjacent vertices have the same color.
Consider an instance~$I$ where agent~$1$ only values the red vertices while agent~$2$ only values the blue vertices, and the value of each vertex~$v$ is $\frac{\deg(v)}{m - 1}$.
Since there are $m-1$ edges, and each edge is adjacent to exactly one red vertex and one blue vertex, each agent's total value is~$1$ and $\utilOPT(I) = 2$.
We claim that $\utilSW(\widehat{\alloc}) \leq \frac{m}{m-1}$ for every connected allocation $\widehat{\alloc} \in C(I)$.
Observe that a connected allocation $\widehat{\alloc}$ must partition the tree into two connected components, with exactly one edge connecting the two components.
Each edge connects a red vertex and a blue vertex, so it adds exactly $\frac{1}{m-1}$ to the utilitarian welfare; the only exception is the edge connecting the two components, which adds at most $\frac{2}{m-1}$.
Hence, we have $\utilSW(\widehat{\alloc}) \leq \frac{m}{m-1}$.
It follows that the utilitarian PoC is at least $\frac{2 (m-1)}{m}$.

\paragraph{Upper Bound:}
Consider an instance~$I$, and let $\Delta_v \coloneqq u_2(v) - u_1(v)$ for each~$v \in V$.
We have
\begin{align*}
\utilOPT(I) = \sum_{v \in V} \max \{u_1(v), u_2(v)\} = \sum_{v \in V} \frac{u_1(v) + u_2(v) + |u_2(v) - u_1(v)|}{2} = 1 + \frac{\sum_{v \in V} |\Delta_v|}{2}.
\end{align*}
Let $\Delta \coloneqq \max_{T} | \sum_{i \in T} \Delta_i|$, where the maximum is taken over all subtrees $T$ such that both $T$ and $G\setminus T$ are connected.
We claim that $\max_{\alloc \in \connectedAllocsOf{I}} \utilSW(\alloc) = 1 + \Delta$.
To see this, observe that the maximum utilitarian welfare under connectivity can be obtained by first allocating all vertices to agent~$1$, and then giving a connected subtree $T$ to agent~$2$ such that $G\setminus T$ remains connected and the sum $\sum_{i \in T} \Delta_i$ is maximized.

In order to show that $\frac{\utilOPT(I)}{\max_{\alloc \in \connectedAllocsOf{I}} \utilSW(\alloc)} \le \frac{2 (m-1)}{m}$, it therefore suffices to show that
\begin{align*}
\frac{2 (m-1)}{m} &\geq \frac{1 + \frac{\sum_{v \in V} |\Delta_v|}{2}}{1 + \Delta}. 
\end{align*}
Upon rearranging, this becomes
\begin{align*}
(2m - 2)(1+\Delta) &\geq m + \frac{m}{2} \sum_{v \in V} |\Delta_v|,
\end{align*}
or equivalently,
\begin{align*}
(2m - 2) \Delta &\geq \frac{m}{2} \left( \sum_{v \in V} |\Delta_v| \right) - (m - 2),
\end{align*}
or
\begin{align*}
\Delta &\geq \frac{\frac{m}{2} \left( \sum_{v \in V} |\Delta_v| \right) - (m - 2)}{2m - 2}.
\end{align*}

It remains to prove the last inequality.
Let $S = \frac{\frac{m}{2} \left( \sum_{v \in V} |\Delta_v| \right) - (m - 2)}{2m - 2}$, and assume for contradiction that $\Delta < S$.
We claim that $|\Delta_v| < \deg(v) \cdot S$ for each vertex~$v \in V$.
Denote by $T_1,\dots,T_{\deg(v)}$ the $\deg(v)$ subtrees emanating from $v$ (where $v$ itself is not included in the subtrees).
By definition of $\Delta$, for each $i \in [\deg(v)]$, it holds that $|\sum_{j \in T_i} \Delta_j| \le \Delta < S$.
Therefore, we have
\begin{align*}
|\Delta_v|
&= \left|\Delta_v + \sum_{i\in[\deg(v)-1]}\sum_{j\in T_i}\Delta_j - \sum_{i\in[\deg(v)-1]}\sum_{j\in T_i}\Delta_j\right| \\
&\le \left|\Delta_v + \sum_{i\in[\deg(v)-1]}\sum_{j\in T_i}\Delta_j\right| + \sum_{i\in[\deg(v)-1]}
\left|\sum_{j\in T_i}\Delta_j\right| \\
&= \left|\sum_{j\in T_{\deg(v)}}\Delta_j\right| + \sum_{i\in[\deg(v)-1]}
\left|\sum_{j\in T_i}\Delta_j\right| \\
&< \deg(v)\cdot S,
\end{align*}
where the first inequality follows from the triangle inequality, and the second equality from the fact that $\sum_{v\in V}\Delta_v = 0$.
Summing this over all vertices~$v \in V$, we get
\begin{align*}
\sum_{v \in V} |\Delta_v| < S \cdot \sum_{v \in V} \deg(v) &= \frac{\frac{m}{2} \left( \sum_{v \in V} |\Delta_v| \right) - (m - 2)}{2m - 2} \cdot 2 (m-1) 
= \frac{m}{2} \left( \sum_{v \in V} |\Delta_v| \right) - (m - 2). 
\end{align*}
This is equivalent to
\begin{align*}
m - 2 &< \frac{m - 2}{2} \sum_{v \in V} |\Delta_v|,
\end{align*}
which implies that $\sum_{v \in V} |\Delta_v| > 2$.
On the other hand, by the triangle inequality, we have $\sum_{v \in V} |\Delta_v| \leq \sum_{v \in V} (u_1(v) + u_2(v)) = 2$.
This yields the desired contradiction.
\end{proof}

By using similar techniques, we can also establish the utilitarian PoC for cycles.

\begin{theorem}
\label{thm:util:2-agent:cycle}
Let $G$ be a cycle.
Then, $\utilPrice(G, 2) = \frac{2k}{k + 1}$, where $k = \floor*{\frac{m}{2}}$.
\end{theorem}

\begin{proof}
We distinguish between odd- and even-length cycles.

\paragraph{Case~1: $m$ is even.}
The cycle contains $m = 2k$ vertices, denoted by $1, 2, \dots, 2k$.
We first show that the utilitarian PoC is at least $\frac{2k}{k+1}$.
Consider an instance with the following valuations:
\begin{itemize}
\item agent~$1$ values odd vertices equally, i.e., $u_1(1) = u_1(3) = \cdots = u_1(2k-1) = 1/k$;
\item agent~$2$ values even vertices equally, i.e., $u_2(2) = u_2(4) = \cdots = u_2(2k) = 1/k$.
\end{itemize}
The optimal utilitarian welfare overall is~$2$, while the optimal utilitarian welfare under connectivity is $\frac{k+1}{k}$, achieved by giving vertex~$1$ to agent~$1$ and the remaining vertices to agent~$2$.
This implies that $\utilPrice(G, 2) \geq \frac{2k}{k + 1}$.

We now show that the PoC is at most~$\frac{2k}{k + 1}$.
Let $\Delta_i \coloneqq u_2(i) - u_1(i)$ for each~$i \in [2k]$.
We have
\begin{align*}
\utilOPT(I) = \sum_{i \in [2k]} \max\{u_1(i), u_2(i)\} = \sum_{i \in [2k]} \frac{u_1(i) + u_2(i) + |u_2(i) - u_1(i)|}{2} = 1 + \frac{\sum_{i \in [2k]} |\Delta_i|}{2}.
\end{align*}
Let $\Delta \coloneqq \max_{P} |\sum_{i \in P} \Delta_i|$, where the maximum is taken over all paths~$P$ in the cycle (including single vertices).
We claim that $\max_{\alloc \in \connectedAllocsOf{I}} \utilSW(\alloc) = 1 + \Delta$.
To see this, observe that the maximum utilitarian welfare under connectivity can be obtained by first allocating all vertices to agent~$1$, and then giving a path~$P$ to agent~$2$ such that the sum $\sum_{i\in P} \Delta_i$ is maximized.

In order to show that $\frac{\utilOPT(I)}{\max_{\alloc \in \connectedAllocsOf{I}} \utilSW(\alloc)} \le \frac{2k}{k+1}$, it therefore suffices to show that
\begin{align*}
\frac{2k}{k + 1} &\geq \frac{1 + \frac{\sum_{i \in [2k]} |\Delta_i|}{2}}{1 + \Delta}.
\end{align*}
Upon rearranging, this becomes
\begin{align*}
2k + 2k \cdot \Delta &\geq k + 1 + \frac{k+1}{2} \sum_{i \in [2k]} |\Delta_i|,
\end{align*}
or equivalently,
\begin{align*}
\Delta &\geq \frac{1 - k + \frac{k+1}{2} \sum_{i \in [2k]} |\Delta_i|}{2k}.
\end{align*}

It remains to prove the last inequality.
Let $S = \frac{1 - k + \frac{k+1}{2} \sum_{i \in [2k]} |\Delta_i|}{2k}$, and assume for contradiction that $\Delta < S$.
In particular, $|\Delta_i| < S$ for all $i \in [2k]$.
We therefore have
\begin{align*}
\sum_{i \in [2k]} |\Delta_i| &< 2k \cdot S = 2k \cdot \frac{1 - k + \frac{k+1}{2} \sum_{i \in [2k]} |\Delta_i|}{2k} = 1 - k + \frac{k+1}{2} \sum_{i \in [2k]} |\Delta_i|.
\end{align*}
This is equivalent to
\begin{align*}
k - 1 &< \frac{k-1}{2} \sum_{i \in [2k]} |\Delta_i|,
\end{align*}
which implies that $\sum_{i \in [2k]} |\Delta_i| > 2$.
On the other hand, by the triangle inequality, we have $\sum_{i \in [2k]} |\Delta_i| \leq \sum_{i \in [2k]} (u_1(i) + u_2(i)) = 2$.
This yields the desired contradiction.

\paragraph{Case~2: $m$ is odd.}
The cycle contains $m = 2k + 1$ vertices, denoted by $1,2,\dots,2k+1$.
The same instance as in Case~1, with the exception that vertex~$2k+1$ is of value~$0$ to both agents, shows that the utilitarian PoC is at least~$\frac{2k}{k+1}$.

We show that the PoC is at most $\frac{2k}{k + 1}$ using a similar argument as in Case~1.
Note that since $2k+1$ is odd, there must exist consecutive $\Delta_i$, $\Delta_{i+1}$ with the same sign (where $\Delta_{2k + 2} = \Delta_1$, and $0$ is considered to have either sign).
Since vertices $i$ and $i+1$ together form a path, we have $|\Delta_i| + |\Delta_{i+1}| = |\Delta_i + \Delta_{i+1}| < \Delta < S$.
Repeating the same argument, we end up with
\begin{align*}
\sum_{i \in [2k+1]} |\Delta_i| &< ((2k + 1) - 1) \cdot S = 1 - k + \frac{k + 1}{2} \sum_{i \in [2k+1]} |\Delta_i|.
\end{align*}
This is equivalent to
\begin{align*}
k - 1 &< \frac{k - 1}{2} \sum_{i \in [2k+1]} |\Delta_i|,
\end{align*}
which implies that $\sum_{i \in [2k+1]} |\Delta_i| > 2$ and leads to the same contradiction as in Case~1.
\end{proof}

\subsection{Any Number of Agents}
Next, we consider the general case where the number of agents can be arbitrary.
We will show that for stars, paths, and cycles, the utilitarian PoC is $\Theta(n)$ for large~$m$.

\begin{proposition}
\label{prop:uswupperbound}
Let $G$ be any graph and $n\ge 2$.
Then, $\utilPrice(G, n) \leq n$.
\end{proposition}

\begin{proof}
Given any instance~$I$, the optimal utilitarian welfare~$\utilOPT(I)$ is at most~$n$, as each agent has a total utility of~$1$ for the entire set of items.
On the other hand, a connected allocation that gives all items to the same agent yields a utilitarian welfare of~$1$.
Hence, the utilitarian PoC is at most~$n$.
\end{proof}

While the trivial upper bound of $n$ holds for any graph, we can improve it for stars, paths, and cycles.

\begin{theorem}
\label{thm:utilitarian-upper-path}
Let $G$ be a path and $n\ge 3$.
Then, $\utilPrice(G, n) \leq n - \frac{1}{nm}$.
\end{theorem}

\begin{proof}
Consider any instance~$I$ with a path~$G$ and $n\ge 3$, and let the items on the path be $g_1,g_2,\dots,g_m$ in this order.
If $\utilOPT(I) \le n - \frac{1}{nm}$, then since a connected allocation that gives all items to the same agent yields a utilitarian welfare of~$1$, we are done.
Thus, assume henceforth that $\utilOPT(I) > n - \frac{1}{nm}$.
In particular, this means that each item can yield utility at least $\frac{1}{nm}$ to at most one agent.

Since agent~$1$ has a total utility of $1$ for all items, there exists an item $g_j$ such that $u_1(g_j) \ge \frac{1}{m}$.
From the previous paragraph, it must hold that $u_2(g_j) < \frac{1}{nm}$.
Let 
\begin{align*}
x_1 = \sum_{z = 1}^{j-1}u_1(g_z),\quad y_1 = \sum_{z = 1}^j u_1(g_z),\quad x_2 = \sum_{z = 1}^{j-1}u_2(g_z), \quad y_2 = \sum_{z = 1}^j u_2(g_z).
\end{align*}
We have $y_1 - x_1 \ge \frac{1}{m}$ and $y_2 - x_2 < \frac{1}{nm}$.
Let $r = \frac{1}{2}\left(\frac{1}{m}-\frac{1}{nm}\right) = \frac{n-1}{2nm}$.
If $|x_2-x_1| < r$ and $|y_2-y_1| < r$, then by the triangle inequality,
\begin{align*}
y_1 - x_1 \le |y_1-y_2|+|y_2-x_2|+|x_2-x_1|
< r + \frac{1}{nm} + r
= \frac{1}{m},
\end{align*}
a contradiction.
Hence, $|x_2-x_1| \ge r$ or $|y_2-y_1| \ge r$.

We claim that in either case, $\max_{\alloc \in \connectedAllocsOf{I}} \utilSW(\alloc) \ge 1 + r$.
Suppose first that $|x_2-x_1| \ge r$.
If $x_2 - x_1 \ge r$, then the connected allocation that gives $g_1,\dots,g_{j-1}$ to agent~$2$ and $g_j,\dots,g_m$ to agent~$1$ yields utilitarian welfare at least $1+r$.
Analogously, if $x_1 - x_2 \ge r$, then the same allocation but with the two agents switched yields utilitarian welfare at least $1+r$.
The case where $|y_2 - y_1| \ge r$ can be handled in a similar manner, with the role of $g_{j-1}$ replaced by $g_j$.
Since $\utilOPT(I) \le n$, in order to show that $\utilPrice(G, n) \leq n - \frac{1}{nm}$, it suffices to show that $\frac{n}{1+r} \le n - \frac{1}{nm}$.

Observe that for any $n\ge 3$, it holds that $n^3 \ge n^2 + 3n - 1$.
This implies that $m(n^3 - n^2 - 2n) \ge n^3 - n^2 - 2n \ge n-1$, that is, $n^3m - n^2m - 2nm - n + 1 \ge 0$.
Therefore,
\begin{align*}
2n^3m^2 
&\le 2n^3m^2 + n^3m - n^2m - 2nm - n + 1 
= (n^2m-1)(2nm+n-1).
\end{align*}
It follows that
\begin{align*}
\frac{n}{1+r} = \frac{n}{1 + \frac{n-1}{2nm}} =
\frac{2n^2m}{2nm+n-1} \le \frac{n^2m-1}{nm} = n - \frac{1}{nm},
\end{align*}
completing the proof.
\end{proof}

Since a path can be obtained from a cycle by removing any edge, we immediately get the following corollary.

\begin{corollary}
Let $G$ be a cycle and $n\ge 3$.
Then, $\utilPrice(G, n) \leq n - \frac{1}{nm}$.  
\end{corollary}

We also establish the same upper bound for stars.

\begin{theorem}
Let $G$ be a star and $n\ge 3$.
Then, $\utilPrice(G, n) \leq n - \frac{1}{nm}$.  
\end{theorem}

\begin{proof}
Consider an instance~$I$ with a star~$G$ and $n\ge 3$.
We will assume that $\utilOPT(I) > n - \frac{1}{nm}$ and show that $\max_{\alloc \in \connectedAllocsOf{I}} \utilSW(\alloc) \ge 1 + r$, where $r = \frac{1}{2}\left(\frac{1}{m}-\frac{1}{nm}\right)$.
As in the proof of \Cref{thm:utilitarian-upper-path}, this is sufficient to imply that $\utilPrice(G, n) \leq n - \frac{1}{nm}$.

Since $\utilOPT(I) > n - \frac{1}{nm}$, each item can yield utility at least $\frac{1}{nm}$ to at most one agent.
Assume without loss of generality that agent~$1$ values the center of the star at most $\frac{1}{nm} < \frac{1}{m}$, so the agent values some leaf~$g$ at least $\frac{1}{m}$.
Therefore, agent~$2$ values $g$ at most $\frac{1}{nm}$.
The connected allocation that gives $g$ to agent~$1$ and the remaining items to agent~$2$ yields utilitarian welfare at least $1 + \frac{1}{m} - \frac{1}{nm} = 1+2r > 1 + r$.
It follows that $\max_{\alloc \in \connectedAllocsOf{I}} \utilSW(\alloc) \ge 1 + r$, as desired.
\end{proof}

Moreover, when $m$ is small relative to $n$, we can obtain improved bounds.
First, for a path with $3$~vertices (which is also a star with $3$ vertices), we present a tight bound.

\begin{proposition}
Let $G$ be a path with $m = 3$~vertices and $n\ge 2$.
Then, $\utilPrice(G, n) = 4/3$.
\end{proposition}

\begin{proof}
The case $n = 2$ follows immediately from \Cref{thm:util:2-agent:tree}, so let $n\ge 3$.

For the lower bound, consider the instance~$I$ where agent~$1$ values the first and third items at $1/2$ each and the second item at~$0$, while agents~$2,\dots,n$ value the second item at~$1$ and the first and third items at~$0$.
It holds that $\utilOPT(I) = 2$ and $\max_{\alloc \in \connectedAllocsOf{I}} \utilSW(\alloc) = 3/2$, and so $\utilPrice(G, n) \ge 4/3$.

For the upper bound, consider any instance~$I$, and let $\alloc$ be an allocation yielding the optimal utilitarian welfare.
If $\alloc$ is connected, we are done; assume therefore that $\alloc$ is not connected.
Without loss of generality, suppose that $\alloc$ gives the first and third items to agent~$1$ and the second item to agent~$2$.
By \Cref{thm:util:2-agent:tree}, there exists a connected allocation $\alloc'$ between agents $1$ and $2$ such that $\frac{\utilSW(\alloc)}{\utilSW(\alloc')} = \frac{\utilOPT(I)}{\utilSW(\alloc')} \le \frac{4}{3}$.
Hence, $\utilPrice(G, n) \le 4/3$.
\end{proof}

If $4\le m\le n+1$, we provide an upper bound of $m-2$ for any graph.
Since $m-2 \le n-1$, this gives a stronger bound than $n - \frac{1}{nm}$ for this special case.

\begin{proposition}
Let $G$ be any graph and $4\le m\le n+1$.    
Then, $\utilPrice(G, n) \le m-2$.
\end{proposition}

\begin{proof}
Consider any instance~$I$, and let $\alloc = (M_1,\dots,M_n)$ be an allocation yielding the optimal utilitarian welfare.
If $\alloc$ gives items to at most two agents, then $\utilSW(\alloc) \le 2$.
Since a connected allocation that gives all items to the same agent yields a utilitarian welfare of~$1$, we have $\utilPrice(G, n) \le 2 \le m-2$.

Next, assume that $\alloc$ gives items to a set $N'\subseteq N$ of at least three agents.
This means that each agent receives at most $m-2$ items in $\alloc$.
By giving each agent $i\in N'$ her most preferred item in~$M_i$, such a (partial) allocation gives the agent a utility of at least
$\frac{u_i(M_i)}{|M_i|} \ge \frac{u_i(M_i)}{m-2}$.
We then extend the partial allocation to a complete, connected allocation in an arbitrary way; this is possible since $G$ is connected.
The resulting allocation has utilitarian welfare at least 
\begin{align*}
\sum_{i\in N'}\frac{u_i(M_i)}{m-2} = \frac{1}{m-2}\sum_{i\in N'} u_i(M_i) = \frac{1}{m-2}\cdot \utilSW(\alloc) = \frac{1}{m-2}\cdot \utilOPT(I).
\end{align*}
Therefore, we again have $\utilPrice(G, n) \le m-2$.
\end{proof}

We now turn to lower bounds, beginning with stars.
In our construction, one agent only values the center vertex, while each remaining agent has a distinct set of leaf vertices that she values equally and the sizes of these sets differ by at most~$1$.

\begin{theorem}
\label{thm:utilstarn}
Let $G$ be a star, $n\ge 2$, and $m=c(n-1)+d+1$, where $c\in \mathbb{Z}_{\ge 1}$ and $d\in \{0,1,\dots,n-2\}$.
Then, $\utilPrice(G, n)\geq \frac{nc(c+1)}{c^2+2nc+n-c-m}$. 
This is $\Omega(n)$ for large $c$ (i.e., large $m$).
\end{theorem}

\begin{proof}
Assume that in the instance~$I$, each item is positively valued by at most one agent, which means that $\utilOPT(I)=n$. 
Let one agent value the center vertex at~$1$, let $d$ agents value $c+1$ of the leaf vertices at $\frac{1}{c+1}$ each, and let $n-1-d$ agents value $c$ of the leaf vertices at $\frac{1}{c}$ each.
This gives
\begin{align*}
\max_{\alloc \in \connectedAllocsOf{I}} \utilSW(\alloc)=1+\frac{d}{c+1}+\frac{n-d-1}{c}
=\frac{c^2+2nc+n-c-m}{c(c+1)},
\end{align*}
where for the last equality we apply the definition of~$m$ from the theorem statement.
Hence, the utilitarian PoC is at least
\[
\frac{nc(c+1)}{c^2+2nc+n-c-m},
\]
as desired.
\end{proof}

We next provide a lower bound for paths, by constructing an instance where the agents value items in cyclic order.

\begin{theorem}
\label{thm: pathuswlb}
Let $G$ be a path, $n\ge 2$, and $m=cn+d$, where $c\in \mathbb{Z}_{\ge 1}$ and $d\in \{0,1,\dots,n-1\}$.
Then, $\utilPrice(G, n)\geq \frac{nc}{n+c-1}$. This is $\Omega(n)$ for large $c$ (i.e., large $m$).
\end{theorem}

\begin{proof}
Assume that in the instance~$I$, each item is positively valued by at most one agent, which means that $\utilOPT(I)=n$. 
Let the agents value the items on the path in the order $(1,2,\dots,n)^*$, where each number in the repeating sequence represents the agent who positively values the corresponding item, and each agent has the same value for all of her valued items.

Since $m = cn+d$, agents $1,2,\dots,d$ positively value $c+1$ items at $\frac{1}{c+1}$ each, while agents $d+1,\dots,n$ positively value $c$ items at $\frac{1}{c}$ each. 
Note that in a connected allocation, in order for an agent to receive each valued item beyond her first valued item, she must receive at least $n$ more items.
Hence, in any connected allocation, at most $n + (c-1)$ items can be allocated to agents who value them---indeed, if at least $n + c$ items are allocated to agents who value them, then by the previous sentence, the total number of items must be at least $n + nc > d + nc = m$, a contradiction.
Since each item has value at most $\frac{1}{c}$, the utilitarian welfare of any connected allocation is at most $\frac{1}{c}(n+c-1)$.
It follows that the utilitarian PoC is at least $\frac{nc}{n+c-1}$.
\end{proof}

Finally, we present a lower bound for cycles.

\begin{theorem}
\label{thm:uswanyncycles}
Let $G$ be a cycle, $n\ge 2$, and $m=cn+d$, where $c\in \mathbb{Z}_{\ge 1}$ and $d\in \{0,1,\dots,n-1\}$.
Then, $\utilPrice(G, n)\geq \frac{nc}{2n+c-2}$. 
This is $\Omega(n)$ for large $c$ (i.e., large $m$).
\end{theorem}

\begin{proof}
We use a similar construction and argument as in \Cref{thm: pathuswlb} for paths.
The difference is that, for cycles, we only argue that in order for an agent to receive each valued item beyond her first \emph{two} valued items, she must receive at least $n$ more items.
This weaker argument is due to the fact that an agent may receive two valued items without receiving $n$ more items if one of those two items is near the beginning of the repeating sequence $(1,2,\dots,n)^*$ while the other is near the end.
Hence, in any connected allocation, at most $2n+(c-2)$ items can be allocated to agents who value them---indeed, if at least $2n + (c-1)$ items are allocated to agents who value them, then by the earlier argument, the total number of items must be at least $2n + n(c-1) = n + nc > d + nc = m$, a contradiction.
Since each item has value at most $\frac{1}{c}$, the utilitarian welfare of any connected allocation is at most $\frac{1}{c}(2n+c-2)$.
It follows that the utilitarian PoC is at least $\frac{nc}{2n+c-2}$.
\end{proof}

\section{Conclusion and Future Work}

In this work, we have introduced the notions of egalitarian and utilitarian price of connectivity (PoC) for the allocation of indivisible items on a graph. 
These notions quantify the worst-case welfare loss due to connectivity constraints, which require each agent to receive a connected subgraph of items. 
We investigated the PoC for various classes of graphs---including graphs with vertex connectivity~$1$ and~$2$, complete bipartite graphs, trees, paths, stars, cycles, as well as complete graphs with a matching removed---and presented tight or asymptotically tight bounds.
Interestingly, in the case of two agents, while the egalitarian PoC and the utilitarian PoC are within a factor of~$2$ for paths, cycles, complete bipartite graphs, and complete graphs with a matching removed, they differ by a factor linear in the number of items for stars.

There is still room to extend the current results further. 
For the case of two agents, while we have derived bounds for classes of dense and sparse graphs, a complete characterization of the egalitarian PoC for arbitrary graphs is not yet known, and the utilitarian PoC for graphs with connectivity~$1$ and $2$ also remains elusive. 
Furthermore, it would be interesting to generalize the existing egalitarian PoC result for trees and three agents (\Cref{thm:egal:tree}) to all graphs with connectivity~$1$, as well as to complement this result with utilitarian PoC bounds.

Finally, as discussed in the survey by \citet{Suksompong21}, in addition to connectivity constraints, other types of constraints have been studied in the fair division literature. 
For instance, budget constraints require the total ``cost'' of items allocated to each agent to be within a given budget~\citep{WuLiGa25}, while matroid constraints restrict the possible bundles received by each agent to the ``independent sets'' of a matroid~\citep{Dror23}.
Investigating the price of such constraints with respect to both egalitarian and utilitarian welfare is a natural direction for future research.

\section*{Acknowledgments}

This work was partially supported by the ARC Laureate Project FL200100204 on ``Trustworthy AI'', by the Singapore Ministry of Education under grant number MOE-T2EP20221-0001, and by an NUS Start-up Grant.
We would like to thank the anonymous IJCAI 2024 and Discrete Applied Mathematics reviewers for their valuable comments.

\bibliographystyle{plainnat}
\bibliography{bibliography}

\end{document}